\documentclass[aps,prd,preprint,superscriptaddress,nofootinbib,longbibliography]{revtex4-2}
\usepackage[T1]{fontenc}
\usepackage[utf8]{inputenc}
\usepackage{lmodern}
\usepackage{microtype}
\usepackage{amsmath,amssymb,amsfonts,amsthm,mathtools}
\usepackage{booktabs}
\usepackage{graphicx}
\graphicspath{{Figure/}}
\usepackage[colorlinks=true,linkcolor=blue,citecolor=magenta,urlcolor=cyan]{hyperref}
\theoremstyle{plain}
\newtheorem{proposition}{Proposition}
\newcommand{\beq}{\begin{equation}}
\newcommand{\eeq}{\end{equation}}
\newcommand{\Neg}{\mathcal{N}}
\newcommand{\NegRate}{\mathcal{R}}
\newcommand{\Wnorm}{\mathcal{W}}
\newcommand{\Psurv}{\mathcal{P}}
\newcommand{\su}{\mathfrak{su}}
\newcommand{\SL}{\mathrm{SL}}
\newcommand{\SU}{\mathrm{SU}}
\newcommand{\AdS}{\mathrm{AdS}}
\newcommand{\sign}{\operatorname{sgn}}
\newcommand{\NB}{\mathrm{NB}}
\newcommand{\lAdS}{\ell_{\AdS}}
\newcommand{\NKrylov}{\hat{N}_{\mathrm{Krylov}}}
\newcommand{\Cplx}{\mathcal{C}}
\begin{document}
\title{The Normalised Wigner Negativity Rate as a Second-Moment Probe of Infall in \texorpdfstring{AdS$_3$}{AdS3}}
\author{Ritam Basu}
\email{ritam.basu@theory.tifr.res.in}
\affiliation{Department of Theoretical Physics, Tata Institute of Fundamental Research, 1 Homi Bhabha Road, Mumbai 400005, India}
\begin{abstract}
In spread complexity, the growth rate of the average position of a time-evolved state along its Krylov chain recovers the
proper radial momentum of an infalling particle in AdS. It is, however, a measure of the first moment only,
irrespective of the spread of the wavepacket away from its classical trajectory. The rate of a
\emph{normalized} Krylov-Wigner negativity can be proposed as a diagnostic of the second moment of the
boundary state that captures this spreading. Starting with the \emph{seed-normalized}
Krylov-Wigner distribution---that is, the Wigner transform of the descendant cloud, with the
decaying return amplitude divided out---we obtain an analytic Bessel form in the macroscopic
limit and compute its total negativity explicitly. Retaining the Bessel variable all the way
through, we find that the negativity goes as $\sinh^{4\Delta}(\pi t/\beta)$. But the raw
normalized-state negativity saturates to an $O(1)$ constant, well below the $O(\sqrt{D})$ bound. Using the
exact negative binomial statistics of the Krylov chain, the normalized negativity is at late
times a fixed power of the second moment of the Krylov wavepacket,
$\Neg(t)\propto\bigl[\mathrm{Var}(\NKrylov)\bigr]^{\Delta}$, for every $\Delta>1/2$. The
relation linearizes precisely at $\Delta=1$. In this regime via the momentum dictionary of
Caputa et al.~\cite{Caputa:2024}, the rate becomes the product of the proper radial position
and momentum, $\NegRate\propto\Cplx\,P_\rho$, which we interpret as the rate of the tidal stretch of nearby geodesics falling into the horizon.
We also comment briefly on a speculative direction for future research, in particular the interpretation of the transverse string size operator in terms of the Krylov number operator through the common $\SU(1,1)$ discrete series.
\end{abstract}
\maketitle
\newpage
\section{Introduction}
\label{sec:intro}
The question of decoding the emergence of classical gravity and geometry from quantum dynamics of boundary conformal field theory (CFT) has been a longstanding problem of gauge-gravity duality \cite{Maldacena:1997re, Witten:1998qj, Gubser:1998bc}. Having noticed that the entanglement entropy is not a suitable characterization of late time geometry behind a black hole horizon, it was conjectured that the growing length of the Einstein-Rosen bridge is dual to the quantum complexity growth of the state on the boundary \cite{Susskind:2014rva, Stanford:2014jda, Brown:2015bva, Brown:2015lvg, Susskind:2018pmk, Balasubramanian:2024ysu}.

The circuit complexity~\cite{Nielsen:2006, Jefferson:2017sdb, Caputa:2017urj,Caputa:2018kdj} depends on the chosen gate set and reference state. The
Krylov basis approach requires no such choices. Here everything is fixed by the Hamiltonian and the initial data. Two closely related notions have
appeared in this context. The Krylov complexity of~\cite{Parker:2018yvk,Barbon:2019wsy, Rabinovici:2020} quantifies the growth of a Heisenberg \emph{operator} under nested commutators. The spread complexity of~\cite{Balasubramanian:2022tpr}, on the other hand, is defined for a \emph{state}. In this case, the Krylov basis is an ordered, orthonormal basis obtained by applying the Gram-Schmidt procedure to the set of states
$\{\psi_0, H\psi_0, H^2\psi_0,\cdots\}$. The spread complexity then measures how spread out the time-evolved state becomes over this basis. It was shown
in~\cite{Balasubramanian:2022tpr} that the Krylov basis minimizes this spreading (see also~\cite{Balasubramanian:2019wgd, Caputa:2021sib,
Bhattacharya:2024}). The literature on both notions has grown considerably by now. We refer the reader to~\cite{Nandy:2024} for a review, and
to~\cite{Rabinovici:2020, Rabinovici:2022} for more details and references. It has been found that a precise dictionary holds in the semiclassical limit between the rate of growth of the spread complexity and the proper radial momentum of an infalling particle in an AdS black hole~\cite{Caputa:2024}.

Nevertheless, the point particle dictionary accounts only for the classical center of mass dynamics. A wavepacket falling into a black hole is never described by a geodesic due to quantum and gravitational effects---it will necessarily \emph{spread}. As a result, the spread complexity, the first moment of the Krylov chain (the average position along the chain), captures the motion of this center of mass but fails to detect higher moments of the probability distribution.

A natural language for questions of this kind is provided by the phase space formulation of quantum mechanics, and it is worth explaining why. The Wigner function $W_\psi(q,p)$ for a state $\psi$ is a quasi-probability distribution which gives a representation of the quantum state in phase space, parametrized by position and momentum simultaneously. This should be contrasted against the wavefunction $\psi(q)$, whose absolute square gives a probability distribution over position alone. The Wigner function is real and unit normalized but it is not a true probability distribution. It can take negative values at some points in phase space. Indeed, had this not been the case, then quantum mechanics would reduce to classical probabilistic dynamics. When the Wigner function for a state is everywhere non-negative, it acts as a genuine probability distribution, and one might roughly expect such states to be ``classical''. Hudson's theorem~\cite{Hudson:1974, Soto:1983} shows that the only states for which the Wigner function is everywhere positive are Gaussian states. Thus, states with everywhere positive Wigner functions form a very restricted class of states in the Hilbert space. But even if we were to start with a classical/Gaussian state, the positivity of the Wigner function will not, in general, be preserved under time evolution. This only happens for at most inhomogeneous, quadratic (i.e., free) Hamiltonians~\cite{SotoEguibar:1983}. The negativity of the Wigner function is a measure of the intrinsic non-classicality in a state. Further, Wigner negativity also gives a lower bound on an entropic measure of spread complexity~\cite{Basu:2025rm} -- a Wigner negative state must unavoidably have some minimal spread, crucially in both position and momentum. The results of~\cite{Balasubramanian:2022tpr} suggest that the Krylov basis is very natural in the sense that the time-dependent wavefunction is as peaked as possible in this basis. This suggests that the state looks as ``classical'' as possible in the Krylov basis, somewhat like coherent states, except that the latter are peaked in phase space. It thus seems natural to study the time-evolution in phase space~\cite{Basu:2024, Basu:2025rm} (see also~\cite{Pal:2026, Camargo:2026, DeRo:2026} for recent work in this direction).

Let us also state precisely what we mean by the Krylov--Wigner negativity. A crucial point is that the discrete Wigner function and its negativity are defined with respect to a choice of basis. A priori, Wigner negativity has nothing to do with the Krylov construction. For any finite dimensional quantum system of odd prime dimension $D$, there exists a natural quasi-probability representation called the discrete Wigner function~\cite{Wootters:1987}. Given an ordered, orthonormal basis for the Hilbert space, the essential idea is to interpret this as the ``position'' basis and correspondingly construct a phase space, the lattice $\mathbb{Z}_D\times\mathbb{Z}_D$, on which the state is represented as a real, unit-normalized quasi-probability distribution. The discrete Wigner function satisfies many of the properties of a standard probability distribution except for one---it can take negative values. Taking the basis to be the Krylov basis (truncated to a large prime dimension $D$) gives the \emph{Krylov--Wigner function} of~\cite{Basu:2024}: the position variable is then the position along the Krylov chain, and its conjugate (momentum) is the discrete Fourier-dual phase. We recall the construction in section~\ref{sec:normalization_upfront}. The amount of negativity in the Wigner function has an operational meaning in the theory of quantum computation. Stabilizer circuits can be simulated efficiently on classical computers and are in this sense essentially classical. Now to achieve universal quantum computation one must supplement them with an additional resource, which Bravyi and Kitaev called ``magic''~\cite{Bravyi:2005}. The total Wigner negativity is a magic monotone, i.e., a measure of magic which is non-increasing on average under stabilizer protocols~\cite{Veitch:2012}. Now by the discrete version of Hudson's theorem, a pure state has an everywhere non-negative discrete Wigner function if and only if it is a stabilizer state~\cite{Gross:2006}. Applied to the Krylov basis, the negativity thus measures how far the spreading wavepacket is from the discrete analog of a classical (Gaussian) state.

Alongside other information-theoretic probes of bulk
geometry~\cite{Banerjee:2023}, we suggest that the normalized generation
rate of the Krylov--Wigner negativity is a useful boundary observable
sensitive to such spreading. The reason it probes the holographic dictionary
via the \emph{second} moment is simple to state. For the $\SU(1,1)$
coherent-state family generated by the local quench, we will show that the
negativity is a fixed power of the variance of the position on the Krylov
chain (see equation~\eqref{eq:general_relation}). Its growth rate therefore
carries exactly the information that the first-moment dictionary
of~\cite{Caputa:2024} discards. We should mention that higher moments and
the full statistics of the position on the Krylov chain have been studied
before, notably in~\cite{Fu:2024, Camargo:2024}. Indeed, the mean--variance
identity~\eqref{eq:var_mean_identity} that we use below already appears
there. What is new here is not the moments themselves, but the phase-space
diagnostic. The negativity rate repackages the second moment into a quantity
with a direct semiclassical reading, together with a proposed bulk
interpretation.

One structural element of our construction is important to state at the
outset. The physically meaningful object is not the raw negativity of the
evolving state (normalized or not), since it saturates when the state
becomes delocalized. Instead, we work with a \emph{normalized} or rescaled
negativity, in which the decay of the return amplitude is divided out. For
the thermal quench considered here, this decay is monotonic and fixed
kinematically by the boundary two-point function. The normalized negativity
and the corresponding rate are defined in sections~\ref{sec:seed_normalized} and~\ref{sec:negativity}.

Our main results are as follows. Results (i)--(iii) are analytical boundary
statements, and do not rely on any bulk conjecture. In words, the logic of
the paper is simple: we first construct the natural phase space distribution
for the spreading cloud, then compute how fast its non-classicality grows,
and finally show that this growth rate directly measures how fast the
wavepacket disperses about its classical trajectory.
\begin{enumerate}
\renewcommand{\labelenumi}{(\roman{enumi})}
\item In section~\ref{sec:seed_normalized}, we consider the \emph{seed-normalized} Krylov--Wigner function, i.e., the Wigner function defined with respect to the Krylov basis, with the decaying return amplitude divided out. The calculation of Wigner negativity in the Krylov basis is facilitated by the fact that, for the local operator quench, the dynamics is controlled by a Hamiltonian in $\su(1,1)$, which turns out to be solvable. In the regime where the excitation has traveled many rungs down the chain, we derive a closed Bessel form for this distribution. We can think of this as a ``continuum'' approximation, where instead of the discrete Krylov lattice, we now have a function of a continuous parameter. The effective phase space has a discrete position variable ranging over half-integer values and a continuous, periodic momentum variable; we will use $Q$ to denote half-integer separated points;
\item In section~\ref{sec:negativity}, we compute the Wigner negativity directly, over the entire phase space, and show that the seed-normalized negativity grows as $\sinh^{4\Delta}(\pi t/\beta)$. On the other hand, the Wigner negativity of the normalized state remains an $O(1)$ constant and does not grow at late times, for every $\Delta>1/2$, which precisely matches the result of~\cite{Balasubramanian:2026wn}. Furthermore, the negativity involves a non-analytic function (namely, the absolute value of the Wigner function). We use the replica trick to deal with these complications. The output of this calculation is an explicit formula for the rate of the Wigner negativity in terms of the return probability;
\item In section~\ref{sec:matching}, we then turn to the second moment of the Krylov wavepacket, and once again derive a universal formula for the Wigner negativity: at late times, the normalized negativity is a fixed power of the second moment of the Krylov wavepacket, $\Neg\propto[\mathrm{Var}(\NKrylov)]^{\Delta}$, for every $\Delta>1/2$. Here $\NKrylov$ is the position operator on the chain, whose mean is the spread complexity. This relation \emph{linearizes} precisely at $\Delta=1$, and only at this dimension does the negativity rate match the growth rate of the variance. In section~\ref{sec:tidal_momentum}, we show that there the rate becomes the product of proper radial position and momentum, $\NegRate\propto\Cplx P_\rho$, which we interpret as the rate of tidal stretching of neighboring infalling geodesics;
\item Finally, in section~\ref{sec:discussion}, we comment briefly on what bulk object could source the negativity in the first place. There is more to AdS/CFT than just Hamiltonian evolution -- one also needs to be able to interpret operator insertions in the boundary and the corresponding matter fields in the bulk within this framework. A structureless point particle cannot: its phase space distribution is Gaussian and positive, and remains so under the quadratic evolution about a geodesic, generating no negativity. String propagation on $\AdS_3$ is described by the $\SL(2,\mathbb{R})$ WZW model, whose discrete series sector realizes the same representation $D^+_\Delta$ of $\SU(1,1)$ -- and hence the same quadratic Casimir $\Delta(\Delta-1)$ -- as the boundary Krylov chain, and the operator measuring the transverse size of the string is then a candidate bulk counterpart of $\NKrylov$. Admittedly, at this stage, this is a rather vague and rudimentary proposal, and there are several issues which need to be addressed in order to make it more concrete; we leave this to future work.
\end{enumerate}
The rest of the paper is organized as follows. In section~\ref{sec:krylov},
we review the necessary background on the Lanczos algorithm, spread
complexity, the local operator quench, and the $\SU(1,1)$ structure of the
thermal Krylov chain. In section~\ref{sec:seed_normalized}, we construct the
seed-normalized Krylov--Wigner distribution and derive its Bessel form in
the macroscopic limit. In section~\ref{sec:negativity}, we compute the
normalized negativity and its rate, and contrast its growth with the
saturation of the raw negativity. In section~\ref{sec:matching}, we compare
the negativity rate with the variance rate and land on the critical
dimension $\Delta=1$. In section~\ref{sec:tidal_momentum}, we give this
matching a bulk reading in terms of tidal momentum. Finally, we end with a
summary, open questions and future directions in
section~\ref{sec:discussion}.
\section{Krylov Basis, Spread Complexity, and the \texorpdfstring{$\SU(1,1)$}{SU(1,1)} Algebra}
\label{sec:krylov}
\subsection{The Krylov basis and spread complexity}
\label{sec:lanczos}
Given an initial state $|\psi_0\rangle$ and Hamiltonian $H$, the Krylov basis is obtained by orthonormalizing the ordered set of states $\{|\psi_0\rangle, H|\psi_0\rangle, H^2|\psi_0\rangle,\cdots\}$ via the Gram-Schmidt process. The construction goes as follows: take $|\psi_0\rangle$ to be the first basis vector. Then consider the state $H|\psi_0\rangle$, and construct the second basis vector by orthogonalizing it with respect to $\psi_0$, and then properly normalizing the result. Then consider the state $H^2\psi_0$, and construct the third basis vector by orthogonalizing it with respect to the first two basis vectors and then normalizing the result, and so on. The resulting basis is called the Krylov basis. Note that the Krylov basis is an ordered basis; we will henceforth use the notation $|K_n\rangle$ to denote the $n$-th Krylov state.

We do not need a basis for the entire Hilbert space, but only for the subspace spanned by the time evolution of the initial state. In~\cite{Basu:2024, Basu:2025rm}, it was important that the Hamiltonian be chaotic, to the extent that time evolution is ergodic; this makes it possible for the Krylov basis to span the entire Hilbert space. The criterion on the initial state being generic is necessary here, because for special states, say for e.g., energy eigenstates, the time evolution becomes trivial multiplication by a phase. When Hamiltonian evolution is ergodic and spans the entire Hilbert space, the dual description presumably involves pure gravity. However, when this fails, we need additional operators (other than the Hamiltonian) in order to construct a basis for the Hilbert space. This should clearly result in a richer structure involving bulk matter.

One of the key features of the Krylov basis is that the Hamiltonian is tri-diagonalized in this basis, i.e.,
\beq
H|K_n\rangle = a_n|K_n\rangle + b_{n+1}|K_{n+1}\rangle + b_n|K_{n-1}\rangle ,
\label{eq:tridiag}
\eeq
with the real coefficients $\{a_n\}$ and $\{b_n\}$ given by:
\beq
a_n = \langle K_n|H|K_n\rangle , \qquad b_n = \langle K_{n-1}|H|K_n\rangle .
\label{eq:lanczos_coeff}
\eeq
These numbers are called the Lanczos or Krylov coefficients.

\paragraph{Spread complexity.}
The Krylov basis has the nice property that it minimizes the ``spreading'' of the wavefunction under time evolution~\cite{Balasubramanian:2022tpr}. This can be quantified as follows: consider the time evolved state
\beq
|\psi(t)\rangle = e^{-iHt}|K_0\rangle = \sum_{n=0}^\infty i^{-n}\psi_n(t)|K_n\rangle .
\label{eq:expansion}
\eeq
Since the Krylov basis tri-diagonalizes the Hamiltonian, the amplitudes $\psi_n(t)$ obey
\beq
\partial_t\psi_n = -i\,a_n\psi_n + b_n\psi_{n-1} - b_{n+1}\psi_{n+1} ,\qquad \psi_n(0)=\delta_{n,0} .
\label{eq:schrodinger}
\eeq
For any ordered, orthonormal basis $\mathcal{B}=\{|n_{\mathcal B}\rangle\}$, we can define a measure of the spread of the state $\psi(t)$:
\beq
C_{\mathcal B}(t)=\sum_{n} c_n\,|\langle n_{\mathcal B}|\psi(t)\rangle|^2 ,
\label{eq:CB_def}
\eeq
where $c_n$ is any positive, monotonically increasing weight function; for e.g., a natural choice is to take $c_n = n$. We may think of $C_{\mathcal B}(t)$ as some notion of ``spread complexity'', i.e., a measure of how spread out the state is in the given choice of basis. Let us say we choose $|0\rangle_{\mathcal B} = |\psi_0\rangle$, so that $C_{\mathcal B}(0)$ is set to zero. This does not fix anything about the remaining basis vectors of $\mathcal B$. But one could ask for more: one could look for a basis where the early growth of $C_{\mathcal B}(t)$ with respect to time is minimized. In other words, we would like a basis which successively minimizes the time derivatives $\{C'_{\mathcal B}(0), C''_{\mathcal B}(0),\cdots\}$. It was shown in~\cite{Balasubramanian:2022tpr} that the Krylov basis minimizes all of these derivatives. In this sense, the Krylov basis quantifies the minimal and inevitable spread in the wavefunction that a state undergoes under quantum evolution.
\beq
\Cplx(t) = \sum_{n=0}^\infty n\,|\psi_n(t)|^2
= \langle\NKrylov\rangle , \qquad
\NKrylov = \sum_{n=0}^\infty n\,|K_n\rangle\langle K_n| .
\label{eq:spread_complexity}
\eeq

\subsection{Local operator quenches in 2d CFT}
\label{sec:quench}
We begin with 2d CFTs, initially in either the vacuum or the thermofield double state on a line excited by a primary operator. Consider a 2d CFT on an infinite line. Let the initial state $|\psi_0\rangle$ be of the form of a local operator quench~\cite{Nozaki:2014, Caputa:2015waa}:
\beq
|\psi_0\rangle = \mathcal{N}_\varepsilon\, e^{-\varepsilon H}\mathcal{O}(x_0)|0\rangle ,
\qquad
|\psi_0\rangle = \mathcal{N}_\varepsilon\, e^{-\varepsilon H_R}\mathcal{O}_R(x_0)\,e^{\varepsilon H_R}\,e^{-\frac{\beta}{2}H_R}|\Omega\rangle ,
\label{eq:local_quench}
\eeq
where $|0\rangle$ is the vacuum and $\mathcal{O}(x_0)$ is a primary operator of conformal dimension $\Delta$, and $|\Omega\rangle$ is the maximally entangled state between two copies (left and right) of the CFT. The evolution over Euclidean time $\tau=\varepsilon$ renders the state normalizable, with finite energy $E\sim\Delta/\varepsilon$.

The return amplitude, $\psi_0(t)=\langle\psi_0|e^{-iHt}|\psi_0\rangle$, can be extracted from the thermal 2-point function~\eqref{eq:survival}. In particular, the Lanczos coefficients can be extracted from it using the moments:
\beq
a_n = \frac{2\pi(n+\Delta)}{\beta\tan\!\left(\tfrac{2\pi\varepsilon}{\beta}\right)} ,
\qquad
b_n = \frac{\pi\sqrt{n(n+2\Delta-1)}}{\beta\sin\!\left(\tfrac{2\pi\varepsilon}{\beta}\right)} .
\label{eq:lanczos_CFT}
\eeq

\subsection{The \texorpdfstring{$\SU(1,1)$}{SU(1,1)} algebra of the Krylov chain}
\label{sec:su11_krylov}
Consider the Lie algebra corresponding to the group $\SL(2,\mathbb{R})$:
\beq
[K_0^{(K)},\,K_\pm^{(K)}] = \pm K_\pm^{(K)} , \qquad
[K_-^{(K)},\,K_+^{(K)}] = 2K_0^{(K)} , \label{eq:su11_comm}
\eeq
and consider a Hamiltonian of the form:
\beq
H=\alpha\bigl(K_+^{(K)}+K_-^{(K)}\bigr)+\gamma K_0^{(K)}+\delta\,\mathbf{1} .
\label{eq:H_sl2}
\eeq
For this choice, the Krylov basis vectors are given by $|K_n\rangle=|\Delta,n\rangle$, where:
\begin{align}
|\Delta,n\rangle&=\sqrt{\frac{\Gamma(2\Delta)}{\Gamma(2\Delta+n)\,n!}}\,\bigl(K_+^{(K)}\bigr)^n|\Delta\rangle , \label{eq:Dn_def}\\
K_0^{(K)}|K_n\rangle &= (n+\Delta)\,|K_n\rangle , \label{eq:K0K}\\
K_+^{(K)}|K_n\rangle &= \sqrt{(n+1)(n+2\Delta)}\,|K_{n+1}\rangle , \label{eq:KpK}\\
K_-^{(K)}|K_n\rangle &= \sqrt{n(n+2\Delta-1)}\,|K_{n-1}\rangle , \label{eq:KmK}\\
C_2^{(K)} &\equiv K_0^{(K)}\bigl(K_0^{(K)}-1\bigr) - K_+^{(K)}K_-^{(K)} = \Delta(\Delta-1)\,\mathbf{1} , \qquad \NKrylov = K_0^{(K)} - \Delta\,\mathbf{1} . \label{eq:casimir_K}
\end{align}
Using the $\SL(2,\mathbb{R})$ algebra, one can show~\cite{Caputa:2021sib, Caputa:2021ori, Balasubramanian:2022tpr} that the Lanczos coefficients in this case are given by:
\beq
a_n=\gamma(\Delta+n)+\delta , \qquad b_n=\alpha\sqrt{n(2\Delta+n-1)}
\label{eq:lanczos_sl2}
\eeq
(see~\cite{Perelomov:1986, Nandy:2024} for more details and references). Comparing with~\eqref{eq:lanczos_CFT}, $\alpha=\pi/[\beta\sin(2\pi\varepsilon/\beta)]$, $\gamma=2\pi/[\beta\tan(2\pi\varepsilon/\beta)]$ and $\delta=0$. The elementary verifications, including the identification of this module as the positive discrete series representation $D^+_\Delta$ of $\SU(1,1)$ with Bargmann index $k=\Delta$ and lowest weight vector $|K_0\rangle$, are recorded in appendix~\ref{app:su11}.

\subsection{The proper momentum correspondence}
\label{sec:conjecture}
In~\cite{Caputa:2024}, it was shown that, for the states~\eqref{eq:local_quench},
\beq
\frac{\mathrm{d}\Cplx}{\mathrm{d}t} \propto P_\rho(t) .
\label{eq:proper_momentum}
\eeq
Holographically, this corresponds to a particle moving on a timelike geodesic of the planar BTZ black hole~\cite{Banados:1992wn, Horowitz:1999jd, Horowitz:2000fv}. The particle is free and therefore falls along a geodesic; this is the setting in which~\eqref{eq:proper_momentum} was established, and the only one we use. Probes with internal structure or spatial extent, whose bulk motion is not that of a structureless particle, have been studied in~\cite{Fatemiabhari:2025}.

The notion of ``classicality'' that is familiar from AdS/CFT is one of large-$N$ factorization, i.e., large $N$ suppresses quantum fluctuations in some simple observables~\cite{Yaffe:1981vf}. On the other hand, Wigner negativity seems like it has more to do with efficiency of sampling the outcome probability distribution of a quantum circuit on a classical computer. It would be good to clarify how these two notions of classicality are related. One way to quantify quantum uncertainty is in terms of the ``spread'' of the wavefunction in position and momentum basis. Summing over one of the directions, say $p$, reduces the Wigner function to the probability density along the other direction (see section~\ref{sec:normalization_upfront}).
\section{Coherent-State Amplitudes and the Seed-Normalized Distribution}
\label{sec:seed_normalized}
\subsection{The \texorpdfstring{$\SU(1,1)$}{SU(1,1)} coherent state and its statistics}
\label{sec:coherent}
The algebra of the previous section largely determines the evolved state. Since $e^{-iHt}$ is an $\SU(1,1)$ group element acting on the lowest weight vector, the time evolved state is a Perelomov coherent state~\cite{Perelomov:1986, Caputa:2021sib, Balasubramanian:2022tpr}, the group theoretic analogue of a squeezed vacuum. Such a state is characterized by a single complex number, the squeezing parameter $A(t)$, and its wavefunction on the $n$-th rung of the chain is
\beq
\psi_n(t) = \sqrt{\frac{\Gamma(2\Delta+n)}{n!\,\Gamma(2\Delta)}}\;A(t)^n\,\psi_0(t) .
\label{eq:coherent_ansatz}
\eeq
The combinatorial prefactor is just the norm of $(K_+^{(K)})^n|K_0\rangle$ in $D_\Delta^+$, so~\eqref{eq:coherent_ansatz} is the exact wavefunction of a coherent state on the chain, with one unknown function $A(t)$ that we determine from the dynamics in section~\ref{sec:A_derivation}. Imposing $\sum_n|\psi_n(t)|^2=1$ and summing the resulting series fixes the seed amplitude,
\beq
|\psi_0(t)|^2=(1-|A(t)|^2)^{2\Delta} ,
\label{eq:psi0_norm}
\eeq
which is real and positive as long as $|A(t)|<1$, the regime of a normalizable coherent state. None of this is new. The closed form amplitudes and the occupation statistics that follow from them were obtained in~\cite{Caputa:2021sib, Balasubramanian:2022tpr}, and the statistics were analyzed in detail in~\cite{Fu:2024, Camargo:2024}. We record them in our conventions because every formula in the later sections is built on them.

The immediate consequence of~\eqref{eq:coherent_ansatz} is that the occupation probabilities $p_n(t)=|\psi_n(t)|^2$ form a one-parameter family controlled by $|A(t)|^2$ alone. This family is a familiar one. Substituting the normalization~\eqref{eq:psi0_norm} and using $\Gamma(2\Delta+n)/[n!\,\Gamma(2\Delta)]=\binom{n+2\Delta-1}{n}$,
\beq
|\psi_n(t)|^2 = \binom{n+2\Delta-1}{n}|A(t)|^{2n}(1-|A(t)|^2)^{2\Delta} ,
\label{eq:NB_dist}
\eeq
which is the negative binomial distribution $\NB(2\Delta,|A(t)|^2)$ (see, e.g.,~\cite{Balasubramanian:2022tpr, Fu:2024}). Its generating function is $G(z)=\sum_n p_n z^n=(1-p)^r(1-pz)^{-r}$ with $r=2\Delta$ and $p=|A|^2$, and differentiating at $z=1$ gives the exact mean and variance,
\begin{align}
\langle\NKrylov\rangle &= \frac{2\Delta|A(t)|^2}{1-|A(t)|^2} ,
\label{eq:nb_mean_lemma}\\
\mathrm{Var}(\NKrylov) &= \frac{2\Delta|A(t)|^2}{(1-|A(t)|^2)^2} .
\label{eq:nb_var_lemma}
\end{align}
The two parameters have a simple meaning. The number of ``failures'' $r=2\Delta$ is twice the Bargmann index of the representation, so the dimension of the quenching operator directly controls the shape of the distribution. The ``success probability'' $p=|A(t)|^2$ measures the squeezing, and increases monotonically towards unity as the state spreads down the chain. As $|A|\to1$ both moments diverge; this is the spreading of the wavepacket characteristic of operator growth in a chaotic theory. (On a finite dimensional chain the spreading is eventually capped at the Heisenberg time, after which the Krylov--Wigner negativity plateaus~\cite{Basu:2024, Basu:2025rm}. The coherent state description below is valid up to that point, in the regime $\langle\NKrylov\rangle\ll D$.)

Two features of~\eqref{eq:NB_dist} drive the rest of the paper. First, the distribution is \emph{overdispersed}: the variance carries an extra factor of $(1-|A|^2)^{-1}$ relative to the mean, so it grows faster as the state delocalizes. Eliminating $|A|^2$ between the two moments gives the exact identity
\beq
\mathrm{Var}(\NKrylov)
= \langle\NKrylov\rangle
+ \frac{1}{2\Delta}\langle\NKrylov\rangle^2 .
\label{eq:var_mean_identity}
\eeq
This identity is not new either; it appears, for instance, as equation (71) of~\cite{Fu:2024} in the same $\SU(1,1)$ setting. We will use it in sections~\ref{sec:tidal_momentum} and~\ref{sec:discussion} to relate the boundary statistics to the bulk. Second, because everything depends on the single combination $|A(t)|^2$, the time dependence of any moment can be read off once $A(t)$ is known. We determine it next; the late time scalings $\langle\NKrylov\rangle\sim\sinh^2(\pi t/\beta)$ and $\mathrm{Var}(\NKrylov)\sim\sinh^4(\pi t/\beta)$ then follow.

\subsection{Determining the squeezing parameter}
\label{sec:A_derivation}
The coherent state ansatz reduces the evolved state to a single unknown function $A(t)$. Naively, fixing it would require solving the full hopping problem~\eqref{eq:schrodinger}. But the coherent state stays coherent under the group evolution, so it is enough to impose the Schr\"odinger equation on a single component; the algebra takes care of the rest. The simplest component is $n=0$. Setting $n=0$ in~\eqref{eq:schrodinger}, using $b_0=0$ and $\psi_{-1}\equiv0$, and reading off $\psi_1=\sqrt{2\Delta}\,A\,\psi_0$ from~\eqref{eq:coherent_ansatz}, we get a first order equation for the seed amplitude alone,
\beq
\partial_t\psi_0 = -i\,a_0\psi_0 - b_1\sqrt{2\Delta}\,A\psi_0 .
\label{eq:n0_eq}
\eeq
Dividing by $\psi_0$,
\beq
\boxed{A(t) = \frac{-\partial_t\ln\psi_0(t) - i\,a_0}{b_1\sqrt{2\Delta}}.}
\label{eq:A_general}
\eeq
So $A(t)$ is determined by one function, the return amplitude $\psi_0(t)=\langle K_0|\psi(t)\rangle$. For a primary of dimension $\Delta$ in a 2d CFT at inverse temperature $\beta$, conformal invariance fixes this to be the regularized thermal two-point function,
\beq
\psi_0(t) = \left[\frac{\sinh\!\left(\tfrac{\pi(t-2i\varepsilon)}{\beta}\right)}
{\sinh\!\left(\tfrac{-2\pi i\varepsilon}{\beta}\right)}\right]^{-2\Delta} ,
\label{eq:survival}
\eeq
where the shift $t\to t-2i\varepsilon$ comes from the Euclidean smearing discussed below~\eqref{eq:local_quench}. Its logarithmic derivative is elementary,
\beq
\partial_t\ln\psi_0(t)
= -\frac{2\pi\Delta}{\beta}\coth\!\left(\frac{\pi(t-2i\varepsilon)}{\beta}\right) .
\label{eq:log_deriv}
\eeq
Substituting this together with $a_0$ and $b_1\sqrt{2\Delta}$ from~\eqref{eq:lanczos_CFT} into~\eqref{eq:A_general}, the common factor $2\pi\Delta/\beta$ cancels and we are left with
\beq
A(t) = \frac{\coth\!\left(\tfrac{\pi(t-2i\varepsilon)}{\beta}\right)
- i\cot\!\left(\tfrac{2\pi\varepsilon}{\beta}\right)}{\csc\!\left(\tfrac{2\pi\varepsilon}{\beta}\right)} .
\label{eq:A_intermediate}
\eeq
The cancellation deserves a comment. The conformal dimension $\Delta$ has dropped out of $A(t)$ entirely; it enters only through the parameter $r=2\Delta$ of the occupation statistics~\eqref{eq:NB_dist}. In other words, the time dependence of the \emph{shape} of the Krylov wavepacket -- and hence the functional form of every scaling derived below -- is universal, while $\Delta$ only sets the relative weights. It is this universality that allows a single matching condition to pick out a critical dimension in section~\ref{sec:matching}.

The remaining step is elementary trigonometry. Writing $a=\pi t/\beta$ and $b=2\pi\varepsilon/\beta$ and putting the two terms in the numerator of~\eqref{eq:A_intermediate} over the common denominator $\sinh(a-ib)\sin b$, the numerator collapses to $-i\sinh a$, and we find
\beq
\boxed{A(t) = \frac{-i\sinh(\pi t/\beta)}{\sinh\!\left(\tfrac{\pi(t-2i\varepsilon)}{\beta}\right)}.}
\label{eq:A_final}
\eeq
So the squeezing parameter is a ratio of thermal sines. Its modulus encodes the spreading of the wavepacket, while the explicit $-i$ and the regulator supply the phase $\varphi_A(t)$, which rotates the Wigner distribution in section~\ref{sec:bessel}. Since $|\sinh(x\pm iy)|^2=\sinh^2 x+\sin^2 y$ does not depend on the sign of $y$, the two forms of $A(t)$ obtained by flipping the sign of the regulator have the same $|A(t)|^2$ and describe the same physics. Explicitly,
\beq
|A(t)|^2=\frac{\sinh^2(\pi t/\beta)}{\sinh^2(\pi t/\beta)+\sin^2(2\pi\varepsilon/\beta)} ,
\label{eq:Asq}
\eeq
so that $1-|A(t)|^2$ is precisely the variable $\xi(t)$ introduced in~\eqref{eq:xi_def} below, and $|A|^2\to1-\sin^2(2\pi\varepsilon/\beta)/\sinh^2(\pi t/\beta)$ at late times. Substituting into~\eqref{eq:nb_mean_lemma}--\eqref{eq:nb_var_lemma} gives the $\sinh^2$ and $\sinh^4$ growth of the mean and the variance quoted above.

\subsection{The discrete Wigner function and the seed-normalized amplitudes}
\label{sec:normalization_upfront}
Before we build the phase space distribution that we actually use, let us recall the standard object of which it will be a rescaling. This material follows~\cite{Wootters:1987, Gross:2006, Basu:2024}, where more details can be found.

For any finite dimensional quantum system of odd prime dimension $D$, there exists a natural quasi-probability representation called the discrete Wigner function~\cite{Wootters:1987}.\footnote{The discrete Wigner function formalism is cleanest when $D$ is an odd prime number, although most of the results hold more generally for odd $D$, and generalizations exist for even $D$ as well. We take $D$ to be a large odd prime throughout, chosen so that truncating the Krylov chain to its first $D$ rungs carries negligible weight; all of our statements are made in the regime $\langle\NKrylov\rangle\ll D$.} Let $\{|k\rangle\}_{k=0}^{D-1}$ be an ordered, orthonormal basis for the Hilbert space. The essential idea is to interpret this as the ``position'' basis, and correspondingly construct a ``phase space''. We define the discrete phase space $\mathcal{P}$ as the lattice $\mathbb{Z}_D\times\mathbb{Z}_D$ of size $D^2$. With respect to the chosen basis, we define a set of $D^2$ operators $A(q,p)$ called phase-point operators, each labelled by a phase space point,
\beq
A(q,p)=\sum_{k,\ell=0}^{D-1}\widehat{\delta}_{2q,k+\ell}\,e^{\frac{2\pi i}{D}(k-\ell)p}\,|k\rangle\langle\ell| ,
\label{eq:phase_point}
\eeq
where the hatted Kronecker delta $\widehat{\delta}$ is the mod $D$ version, i.e., it is one when $(k+\ell)=2q \bmod D$, and zero otherwise. The discrete Wigner function for a density matrix $\rho$ is then defined as
\beq
W_\rho(q,p)=\frac{1}{D}\mathrm{Tr}\bigl(\rho\,A(q,p)\bigr) .
\label{eq:W_def}
\eeq
The Wigner function should be thought of as an attempt to represent the quantum state $\rho$ as a probability distribution in phase space, as would be the case in classical mechanics. Indeed, it is real and unit normalized, $\sum_{q,p}W_\rho(q,p)=\mathrm{Tr}\,\rho=1$; summing over one of the directions, say $p$, reduces it to the probability density along the other direction, $\sum_p W_\rho(q,p)=\langle q|\rho|q\rangle$; and its time evolution follows a discrete version of the Moyal equation, the quantum analogue of the Liouville equation from classical mechanics. This interpretation fails for one important reason: the Wigner function can take negative values at some points in phase space. Indeed, had this not been the case, quantum mechanics would reduce to classical probabilistic dynamics. For odd prime $D$, a pure state has an everywhere non-negative Wigner function if and only if it is a stabilizer state; this is the discrete version of Hudson's theorem~\cite{Gross:2006}. The negativity of the Wigner function will play a central role in this work.

We now take the basis to be the Krylov basis of section~\ref{sec:krylov}, truncated to its first $D$ rungs. The resulting object -- the Wigner function of the time evolved state $|\psi(t)\rangle$ with respect to the Krylov basis -- is what we called the \emph{Krylov--Wigner function} in~\cite{Basu:2024}. Writing $\rho=|\psi(t)\rangle\langle\psi(t)|$ and labelling the pair of chain sites in~\eqref{eq:phase_point} by their midpoint and half-separation, $k=Q+m$ and $\ell=Q-m$, the definition~\eqref{eq:W_def} becomes
\beq
W(Q,p,t)=\frac{1}{D}\sum_{m=-Q}^{Q}e^{4\pi i m p/D}\,\psi_{Q+m}(t)\,\psi^*_{Q-m}(t) .
\label{eq:wootters_def}
\eeq
Here we have represented the Wigner function on the lattice of half-integer separated points $Q\in\{0,\tfrac12,1,\tfrac32,\dots\}$, with the understanding that $Q=\tfrac12$ corresponds to $q=(D+1)/2$ in $\mathbb{Z}_D$, and so on. This is merely a re-ordering of the phase-point operators, and a convenient one, because with respect to this ordering the spread of the Wigner function along the chain has a local appearance~\cite{Basu:2024}. The mod $D$ wrap-around in the sum plays no role in the regime $\langle\NKrylov\rangle\ll D$, and we drop it.

It is worth pausing to explain the meaning of the two phase space variables, since neither is the position or momentum of a particle in continuous space. The midpoint $Q$ tracks the position along the Krylov chain; it is our discrete analogue of the radial coordinate, and its moments are the moments of $\NKrylov$. The variable $p\in\{0,1,\dots,D-1\}$ is the discrete Fourier conjugate of the site separation $k-\ell=2m$. It is sensitive to the \emph{phases} of the coherences $\psi_k\psi_\ell^*$ between chain sites, and plays the role of the momentum canonically conjugate to the chain position. (The phase $i^{-n}$ in the expansion~\eqref{eq:expansion} contributes $i^{-2m}=e^{-i\pi m}$ to the kernel, i.e., a rigid shift of $p$ by $D/4$; since all our observables are summed over the whole phase space, we drop it.) The Krylov--Wigner function therefore gives us a way of visualizing how the state $|\psi(t)\rangle\langle\psi(t)|$ evolves over the discrete Krylov phase space, with the chain position and its conjugate phase treated on the same footing.

Now to the normalization, which determines what the negativity actually measures. The evolved state is always normalized, $\sum_n|\psi_n(t)|^2=1$, for any $t$ and any $\beta$, by unitarity. Equivalently, the \emph{return probability}
\beq
\Psurv(t)\equiv|\psi_0(t)|^2=(1-|A(t)|^2)^{2\Delta}
\label{eq:survival_prob}
\eeq
ensures that $\sum_n|\psi_n|^2=|\psi_0|^2\sum_n\binom{n+2\Delta-1}{n}|A|^{2n}=1$ identically, with no distinguished value of $\beta$ at which normalization fails. The Wigner function of the state correspondingly satisfies $\sum_{Q,p}W=1$ at all times. The trouble is that, as we will see in section~\ref{sec:growth}, the negativity of this normalized distribution saturates: the quasi-probability spreads over an ever larger region of phase space while its height decreases, and the total absolute weight stops growing once the state has delocalized. That saturation is a statement about the decay of the return amplitude, not about the geometric spreading we want to capture. To expose the spreading, we normalize the amplitudes relative to the surviving seed from the start. Define the \emph{seed-normalized} amplitudes
\beq
\chi_n(t)\equiv\frac{\psi_n(t)}{\psi_0(t)}
=\sqrt{\frac{\Gamma(2\Delta+n)}{n!\,\Gamma(2\Delta)}}\,A(t)^n ,\qquad
\chi_0(t)=1.
\label{eq:chi_def}
\eeq
These are the amplitudes of the descendant cloud \emph{relative to the surviving seed}. Note carefully what has been removed: only the $n$-independent magnitude and phase of the return amplitude $\psi_0(t)$. The relative phases of the cloud are untouched, since $A(t)=|A(t)|e^{i\varphi_A}$ enters $\chi_n$ as $A(t)^n$, and it is precisely these relative phases that rotate the Wigner distribution in section~\ref{sec:bessel}. What remains is the full shape of the spreading packet, magnitudes and phases alike. The norm of the $\chi_n$ is the inverse return probability,
\beq
\sum_{n=0}^\infty|\chi_n(t)|^2
=\frac{1}{\Psurv(t)}=(1-|A(t)|^2)^{-2\Delta} ,
\label{eq:chi_norm}
\eeq
which diverges as the state delocalizes and can be thought of as an effective number of visited descendants.

The \emph{seed-normalized Krylov--Wigner distribution} is the discrete Wigner function~\eqref{eq:wootters_def} evaluated on the $\chi_n$,
\beq
\Wnorm(Q,p,t)
\equiv\frac{1}{D}\sum_{m=-Q}^{Q}e^{4\pi i m p/D}\,\chi_{Q+m}(t)\,\chi^*_{Q-m}(t)
=\frac{W(Q,p,t)}{\Psurv(t)} ,
\label{eq:Wnorm_def}
\eeq
where $W$ is the Wigner function of the normalized state. Its total quasi-probability is
\beq
\sum_{Q,p}\Wnorm(Q,p,t)=\frac{1}{\Psurv(t)} .
\label{eq:Wnorm_total}
\eeq
Equivalently, $\Wnorm$ is the Wigner transform of the unnormalized operator $|\psi(t)\rangle\langle\psi(t)|/\Psurv(t)$.

\paragraph{What $\Wnorm$ is, and is not.}
Let us be explicit about this, since it fixes the physical content of everything that follows. $\Wnorm$ is \emph{not} the Wigner function of a normalized state: its total weight~\eqref{eq:Wnorm_total} is $1/\Psurv(t)\neq1$. It is a rescaled phase space quantity, and the negativity computed from it in section~\ref{sec:negativity} is the raw (normalized-state) negativity divided by the return probability, i.e., the magic content per unit surviving return amplitude. The rescaling is uniform over phase space; it is not a projection onto a subspace, it simply compensates for the decay of the return amplitude. Two further clarifications are in order. First, about the word ``trivial'', which we used for this decay in an earlier version of this paper. In general a return amplitude can have recurrences and oscillations, and then no part of its time dependence deserves that name. For the thermal quench at hand, however, the return probability is fixed kinematically by the two-point function~\eqref{eq:survival}, $\Psurv(t)\propto[\sinh^2(\pi t/\beta)+\sin^2(2\pi\varepsilon/\beta)]^{-2\Delta}$, and decays monotonically; recurrences set in only at times exponentially large in the effective Hilbert space dimension, far beyond the window in which we work. It is this smooth decay -- which is blind to how the descendant cloud is organized internally -- that we divide out. In a system with an oscillating return amplitude the definition~\eqref{eq:Wnorm_def} still makes sense pointwise in $t$, but reading $\Neg$ as a monotone diagnostic would require more care. Second, about what is gained over the usual Wigner negativity. The raw negativity saturates at an $O(1)$ plateau (Proposition~\ref{prop:plateau} below), and is therefore blind at late times to the continued growth of the width of the wavepacket. The seed-normalized $\Neg$ is the rescaling under which the phase space diagnostic keeps registering the second moment, cf.~\eqref{eq:general_relation}. We come back to what this choice means for the matching problem at the end of section~\ref{sec:growth}.

\subsection{Bessel form of the seed-normalized distribution}
\label{sec:bessel}
We now derive a closed form for $\Wnorm$ in the regime $Q\gg1$, which we refer to as \emph{macroscopic}. What we mean by this is the following. Once the excitation has traveled many rungs down the chain, the discreteness of the lattice is a subleading effect, and the sums over chain sites are well approximated by integrals, with controlled corrections in $1/Q$ (appendices~\ref{app:stirling} and~\ref{app:euler}). This is also the window in which the bulk reading applies: the semiclassical dictionary of~\cite{Caputa:2024} holds when $\langle\NKrylov\rangle\sim\sinh^2(\pi t/\beta)$ is large, i.e., when the wavepacket sits at macroscopic $Q$, and it is there -- always keeping $\langle\NKrylov\rangle\ll D$ -- that our late time statements are made. Said differently, we are interested in the limits $D\to\infty$ followed by $t\to\infty$, in that order.

Since $\chi_n$ differs from $\psi_n$ only by the common factor $\psi_0(t)$, the computation is that of the Wigner function of an $\SU(1,1)$ coherent state with the survival factor stripped off. Substituting~\eqref{eq:chi_def} into~\eqref{eq:Wnorm_def} and writing $A(t)=|A(t)|e^{i\varphi_A}$, the phases combine into a single exponent and the magnitudes into a ratio of Gamma functions,
\beq
\Wnorm(Q,p,t)=\frac{|A(t)|^{2Q}}{D\,\Gamma(2\Delta)}
\sum_{m=-Q}^{Q}e^{im(4\pi p/D+2\varphi_A)}\,
\frac{\sqrt{\Gamma(2\Delta+Q+m)\Gamma(2\Delta+Q-m)}}{\sqrt{(Q+m)!(Q-m)!}} .
\label{eq:Wnorm_discrete}
\eeq
The prefactor $|A|^{2Q}$ is inherited from the geometric series structure of the amplitudes~\eqref{eq:chi_def}: since $|\chi_n|\propto|A|^n$, a pair of sites with midpoint $Q$ always contributes $|A|^{2Q}$. It is positive everywhere and sets the envelope of the distribution in the radial variable $Q$. All the oscillations in phase space come from the sum, where the kernel $K(Q,m)$ (the ratio of Gamma functions) fixes the profile in the shift variable $m$, and the exponential carries the dependence on the conjugate momentum through the combination $4\pi p/D+2\varphi_A$.

Two standard steps turn the sum into something we can integrate. First, Stirling's approximation (appendix~\ref{app:stirling}) reduces the kernel to a power law,
\beq
K(Q,m)\approx Q^{2\Delta-1}\bigl(1-x^2\bigr)^{\Delta-1/2},\qquad x\equiv m/Q\in[-1,1] .
\label{eq:stirling_kernel}
\eeq
Second, since $\Delta>1/2$ the profile vanishes at $x=\pm1$, so the Euler--Maclaurin formula (appendix~\ref{app:euler}) converts the sum into an integral with no boundary terms,
\[
Q^{2\Delta}\!\int_{-1}^{1}\!\bigl(1-x^2\bigr)^{\Delta-1/2}e^{iZx}\,\mathrm{d}x ,
\qquad
Z\equiv Q\!\left(\frac{4\pi p}{D}+2\varphi_A\right),
\]
and the single rescaled variable $Z$ now carries the entire phase space dependence. The integral is Poisson's representation of the Bessel function: for any $\Delta>-\tfrac12$,
\beq
\int_{-1}^{1}\bigl(1-x^2\bigr)^{\Delta-1/2}e^{iZx}\,\mathrm{d}x
=\sqrt{\pi}\,\Gamma\!\left(\Delta+\tfrac12\right)\!\left(\frac{2}{Z}\right)^{\!\Delta}J_\Delta(Z) .
\label{eq:poisson_bessel}
\eeq
Combining this with the prefactor in~\eqref{eq:Wnorm_discrete} and using the Legendre duplication formula $\Gamma(2\Delta)=2^{2\Delta-1}\pi^{-1/2}\Gamma(\Delta)\Gamma(\Delta+\tfrac12)$ to collapse the Gamma functions and powers of two, we arrive at the leading macroscopic form of the seed-normalized distribution,
\beq
\boxed{\Wnorm(Q,p,t)\approx
\frac{\pi\,|A(t)|^{2Q}\,Q^{2\Delta}}{D\,2^{\Delta-1}\,\Gamma(\Delta)\,Z^\Delta}\,J_\Delta(Z) .}
\label{eq:wigner_bessel}
\eeq
Equation~\eqref{eq:wigner_bessel} is the leading term in the macroscopic (continuum) expansion, not an exact finite chain result. It is cleaner than the Wigner function of the normalized state in one respect: there is no survival factor, since $\Psurv(t)=|\psi_0|^2$ was divided out at the level of the amplitudes~\eqref{eq:chi_def}. The positive envelope $|A|^{2Q}Q^{2\Delta}$ sets the radial extent, while the oscillatory factor $J_\Delta(Z)/Z^\Delta$ -- and in particular its sequence of sign-changing zeros -- produces alternating troughs of negative quasi-probability whose total we compute next. Such troughs are absent for any Gaussian (stabilizer) state. They are the non-classical content of the state, and they tell us that the infalling excitation is not a structureless point particle.
\section{Normalized Negativity and Its Rate}
\label{sec:negativity}
\subsection{Wigner negativity as a magic monotone}
\label{sec:magic}
Throughout this section, the phase space is the discrete lattice $\mathcal{P}=\mathbb{Z}_D\times\mathbb{Z}_D$ of section~\ref{sec:normalization_upfront}, with $Q$ the (midpoint) position along the Krylov chain and $p$ its discrete conjugate. Let us also make precise the sense in which we use the word ``magic'', which we have so far used loosely. Circuits which only involve Wigner positive elements -- stabilizer states, Clifford unitaries and measurements in the computational basis -- evolve within a very small subset of states in the Hilbert space, and can be simulated efficiently on a classical computer; this is the content of the Gottesman--Knill theorem and its phase space generalization~\cite{Veitch:2012}. Such circuit models cannot be universal. In order to achieve universal quantum computation, one needs an additional resource, which Bravyi and Kitaev called magic~\cite{Bravyi:2005}: a state is said to be magic if it is not a stabilizer state. Much like entanglement plays the role of a resource in the theory of quantum communication, magic plays the role of a resource in the theory of quantum computation, with stabilizer states and operations being treated as cheap. To quantify the notion, one defines a \emph{magic monotone} as any real-valued, non-negative function on density matrices which vanishes on stabilizer states and is, on average, non-increasing under any stabilizer protocol~\cite{Veitch:2012}. It turns out that the total negativity of the Wigner function,
\beq
\Neg_{\mathrm{raw}}(t)\equiv\sum_{Q,p}\bigl|W(Q,p,t)\bigr| ,
\label{eq:raw_neg_def}
\eeq
is a magic monotone~\cite{Veitch:2012}. It is evident from the definition that $\Neg_{\mathrm{raw}}\ge1$, with equality if and only if the Wigner function is nowhere negative, i.e., by the discrete Hudson theorem~\cite{Gross:2006}, if and only if the pure state is a stabilizer state. It is also bounded above by $\sqrt{D}$: the orthogonality of the phase-point operators gives $\sum_{Q,p}W^2=\tfrac{1}{D}\mathrm{Tr}\,\rho^2$, and Jensen's inequality then yields $\sum_{Q,p}|W|\le\sqrt{D}\sqrt{\mathrm{Tr}\,\rho^2}\le\sqrt{D}$~\cite{Basu:2024}. In this sense one may regard the Wigner negativity as a measure of stabilizer complexity, or, more loosely, of how far the state is from being classical in the given basis.

For the descendant cloud, the natural object is the negativity of the seed-normalized distribution, its phase space $1$-norm,
\beq
\Neg(t)\equiv\sum_{Q,p}\bigl|\Wnorm(Q,p,t)\bigr|
=\frac{1}{\Psurv(t)}\sum_{Q,p}\bigl|W(Q,p,t)\bigr|
=\frac{\Neg_{\mathrm{raw}}(t)}{\Psurv(t)} .
\label{eq:negativity_def}
\eeq
Here $\Neg_{\mathrm{raw}}$ is the negativity of the normalized state, which obeys the bounds above and, as we show in section~\ref{sec:growth}, saturates at an $O(1)$, $\Delta$-dependent constant within the macroscopic description. Thus $\Neg(t)$ is the magic content per unit surviving return amplitude; dividing by $\Psurv$ removes the kinematic decay of the return amplitude and leaves the growth.

\subsection{The rate via a replica functional}
\label{sec:replica}
We will want the time derivative of $\Neg(t)$. Differentiating is slightly delicate, because $|\Wnorm|$ has a kink wherever $\Wnorm$ vanishes, i.e., on the moving zero set of $J_\Delta(Z)$. We deal with this by a replica trick, much as in~\cite{Basu:2025rm}: for $k>1/2$, define the real, non-negative \emph{replica functional}
\beq
\Neg^{(k)}(t)=\sum_{Q,p}\bigl(\Wnorm(Q,p,t)^2\bigr)^k ,
\qquad \lim_{k\to1/2}\Neg^{(k)}(t)=\Neg(t) .
\label{eq:replica_functional}
\eeq
For $k>1/2$ the summand is a smooth function of $\Wnorm$, so $\Neg^{(k)}$ can be differentiated term by term, and sending $k\to\tfrac12$ afterwards restores the absolute value.
\begin{proposition}[Negativity rate]
\label{prop:neg_rate}
The rate of the seed-normalized negativity is
\beq
\NegRate(t)\equiv\dot{\Neg}(t)
=\sum_{Q,p}\sign\!\bigl(\Wnorm\bigr)\,\partial_t\Wnorm
=\frac{\mathrm{d}}{\mathrm{d}t}\!\left(\frac{1}{\Psurv(t)}\right)
-2\!\sum_{\Omega_-}\partial_t\Wnorm(Q,p,t) ,
\label{eq:neg_rate_final}
\eeq
where $\Omega_-=\{(Q,p):\Wnorm<0\}$. When the distribution has conserved total weight -- as does the normalized-state $W$, for which $\sum_{Q,p}W=1$ -- the first term is absent and the rate is the pure inward flux $\NegRate_{\mathrm{raw}}=-2\sum_{\Omega_-}\partial_t W$.
\end{proposition}
\noindent\textit{Proof.} Differentiating~\eqref{eq:replica_functional} term by term,
$\partial_t(\Wnorm^2)^k=2k(\Wnorm^2)^{k-1}\Wnorm\,\partial_t\Wnorm\to\sign(\Wnorm)\,\partial_t\Wnorm$
as $k\to\tfrac12$ at every point with $\Wnorm\neq0$; the zero set has measure zero and $\partial_t\Wnorm$ is finite there. This gives the first equality. Splitting phase space into $\Omega_\pm$, where $\Wnorm\gtrless0$, and writing $\sum_{\Omega_+}\partial_t\Wnorm=\sum_{Q,p}\partial_t\Wnorm-\sum_{\Omega_-}\partial_t\Wnorm$,
\[
\NegRate=\sum_{\Omega_+}\partial_t\Wnorm-\sum_{\Omega_-}\partial_t\Wnorm
=\sum_{Q,p}\partial_t\Wnorm-2\sum_{\Omega_-}\partial_t\Wnorm .
\]
Finally $\sum_{Q,p}\partial_t\Wnorm=\partial_t\sum_{Q,p}\Wnorm=\partial_t(1/\Psurv)$ by~\eqref{eq:Wnorm_total}, which vanishes precisely when the total weight is conserved. $\square$

Equation~\eqref{eq:neg_rate_final} has a transparent reading. For a trace-preserving distribution, the negativity changes only through quasi-probability flowing across the zero locus into the negative cells. For the seed-normalized $\Wnorm$ there is, in addition, the rescaling term $\partial_t(1/\Psurv)$, which accounts for the growth of the descendant weight. In practice we will obtain the rate by differentiating a closed form for $\Neg(t)$, which we now compute.

\subsection{Phase space integration}
\label{sec:phasespace_integration}
We evaluate~\eqref{eq:negativity_def} directly from the Bessel form~\eqref{eq:wigner_bessel}, \emph{retaining the Bessel variable $Z=Q\tilde\theta$ throughout}, with $\tilde\theta\equiv4\pi p/D+2\varphi_A$. Some care is needed with the discrete-to-continuum measure: three factors of two arise, each with a definite lattice origin, and we record them explicitly since dropping any one of them changes the overall coefficient.

\paragraph{(i) The momentum sum and the mirror copy.}
As $p$ runs over $\{0,1,\dots,D-1\}$, the angle $\tilde\theta=4\pi p/D+2\varphi_A$ sweeps an interval of total length $4\pi$ -- \emph{two} full periods of the phase $e^{im\tilde\theta}$ -- because of the doubled frequency $4\pi/D$ inherited from the midpoint parametrization $k-\ell=2m$. The Bessel profile $|J_\Delta(Q\tilde\theta)|/|Q\tilde\theta|^\Delta$ is localized about $\tilde\theta\equiv0\ (\mathrm{mod}\ 2\pi)$, so the momentum sum picks up two mirror copies of the distribution, centered at lattice momenta separated by $D/2$. (This doubling is directly visible in the exact discrete Wigner function: for a bulk row $Q$, the two maxima of $|W(Q,p)|$ over $p$ sit half a period apart.) Hence
\beq
\sum_{p=0}^{D-1}\;\longrightarrow\;\frac{D}{4\pi}\int_0^{4\pi}\mathrm{d}\tilde\theta
\;=\;2\times\frac{D}{4\pi}\int_{\mathrm{one\ period}}\mathrm{d}\tilde\theta .
\label{eq:momentum_measure}
\eeq
\paragraph{(ii) The angular integral.}
Within one period, the change of variables $Z=Q\tilde\theta$ ($\mathrm{d}\tilde\theta=\mathrm{d}Z/Q$) supplies an explicit factor of $1/Q$, and the envelope is even in $Z$, decaying on both sides of the peak:
\beq
\int_{\mathrm{one\ period}}\mathrm{d}\tilde\theta\,\frac{|J_\Delta(Z)|}{|Z|^\Delta}
=\frac1Q\int_{-\infty}^{\infty}\mathrm{d}Z\,\frac{|J_\Delta(Z)|}{|Z|^\Delta}
=\frac{2c_\Delta}{Q},
\qquad
c_\Delta\equiv\int_0^\infty\frac{|J_\Delta(Z)|}{Z^\Delta}\,\mathrm{d}Z .
\label{eq:cDelta}
\eeq
Here $c_\Delta<\infty$ for every $\Delta>1/2$: the integrand is regular at $Z=0$, where it tends to $2^{-\Delta}/\Gamma(\Delta+1)$, and decays as $Z^{-\Delta-1/2}$ at large $Z$. Two remarks. First, the angular integral is manifestly \emph{not} $Q$-independent; had we used the large-$Z$ (Hankel) form alone, we would have found a spurious divergence at $\tilde\theta\to0$, which is exactly the region $Z\lesssim1$ where that form is invalid and which supplies the $1/Q$. Second, this factor reduces the radial power from $Q^{2\Delta}$ to $Q^{2\Delta-1}$, and it is this reduction that is responsible for the saturation of the raw negativity derived below.
\paragraph{(iii) The radial sum.}
The midpoint $Q=\tfrac12(k+\ell)$ lives on the half-integer lattice $\{0,\tfrac12,1,\tfrac32,\dots\}$ with spacing $\tfrac12$, so
\beq
\sum_{Q}\;\longrightarrow\;2\int_0^\infty\mathrm{d}Q .
\label{eq:radial_measure}
\eeq
\paragraph{Assembly.}
Combining~\eqref{eq:momentum_measure}--\eqref{eq:radial_measure} with the Bessel form~\eqref{eq:wigner_bessel} and writing $|A|^{2Q}=e^{-2Q|\ln|A||}$,
\begin{align}
\Neg(t)
&\approx
2\int_0^\infty\!\mathrm{d}Q\;\Bigl(\frac{D}{4\pi}\Bigr)\,
\frac{\pi\,Q^{2\Delta}}{D\,2^{\Delta-1}\Gamma(\Delta)}\,
e^{-2Q|\ln|A||}\times\frac{4c_\Delta}{Q}
\notag\\[2pt]
&=\frac{2c_\Delta}{2^{\Delta-1}\Gamma(\Delta)}
\int_0^\infty\mathrm{d}Q\,Q^{2\Delta-1}e^{-2Q|\ln|A||}
=\frac{2c_\Delta}{2^{\Delta-1}\Gamma(\Delta)}\,
\frac{\Gamma(2\Delta)}{\bigl(2|\ln|A||\bigr)^{2\Delta}} ,
\label{eq:neg_start}
\end{align}
so that
\beq
\boxed{\;\Neg(t)\approx\frac{c_\Delta\,\Gamma(2\Delta)}{2^{\Delta-2}\,\Gamma(\Delta)}\,
\frac{1}{\bigl(2|\ln|A(t)||\bigr)^{2\Delta}}
\;=\;\frac{2^{\Delta+1}}{\sqrt{\pi}}\,\Gamma\!\Bigl(\Delta+\tfrac12\Bigr)\,c_\Delta\,
\frac{1}{\bigl(2|\ln|A(t)||\bigr)^{2\Delta}}\; ,}
\label{eq:neg_full}
\eeq
where the second form follows from Legendre duplication. Because the survival factor was removed at the level of the amplitudes~\eqref{eq:chi_def}, no cancellation occurs here, and the seed-normalized negativity is a growing function of time. We contrast this with the raw negativity below, which also serves as a consistency check on the coefficient.

\subsection{Growth of the normalized negativity and saturation of the raw negativity}
\label{sec:growth}
Introduce the late time variable
\beq
\varepsilon_0\equiv\sin^2\!\left(\tfrac{2\pi\varepsilon}{\beta}\right),\qquad
\xi(t)\equiv1-|A(t)|^2=\frac{\varepsilon_0}{S^2+\varepsilon_0},\qquad
S\equiv\sinh(\pi t/\beta),
\label{eq:xi_def}
\eeq
so that $2|\ln|A||=-\ln(1-\xi)\simeq\xi\simeq\varepsilon_0/S^2$ at late times. The seed-normalized negativity~\eqref{eq:neg_full} then grows as
\beq
\boxed{\;\Neg(t)\;\simeq\;\frac{c_\Delta\,\Gamma(2\Delta)}{2^{\Delta-2}\Gamma(\Delta)}\,
\frac{\sinh^{4\Delta}(\pi t/\beta)}{\sin^{4\Delta}(2\pi\varepsilon/\beta)}\; ,}
\label{eq:neg_growth}
\eeq
with rate
\beq
\boxed{\;\NegRate(t)\;\simeq\;\frac{c_\Delta\,\Gamma(2\Delta)}{2^{\Delta-2}\Gamma(\Delta)}\,
\frac{4\Delta\pi}{\beta}\,
\frac{\sinh^{4\Delta-1}(\pi t/\beta)\cosh(\pi t/\beta)}{\sin^{4\Delta}(2\pi\varepsilon/\beta)}\; .}
\label{eq:neg_rate_growth}
\eeq
Consistency with the normalized-state picture is immediate. By~\eqref{eq:negativity_def} the raw negativity is $\Neg_{\mathrm{raw}}(t)=\Psurv(t)\,\Neg(t)$, and with $\Psurv=|\psi_0|^2=\xi^{2\Delta}$,
\beq
\Neg_{\mathrm{raw}}(t)
=\frac{c_\Delta\Gamma(2\Delta)}{2^{\Delta-2}\Gamma(\Delta)}
\left(\frac{\xi}{-\ln(1-\xi)}\right)^{2\Delta}
\;\xrightarrow[\;\xi\to0\;]{}\;\frac{c_\Delta\Gamma(2\Delta)}{2^{\Delta-2}\Gamma(\Delta)} .
\label{eq:raw_saturates}
\eeq
So the raw, normalized-state negativity \emph{saturates} to a constant, in agreement with the $O(1)$ late time constancy of the state Wigner negativity for 2d CFT quenches found in~\cite{Balasubramanian:2026wn}. The growth~\eqref{eq:neg_growth} is carried entirely by the seed normalization, as intended. Numerically, the plateau value is $\Neg_\infty\approx3.22,\,2.47,\,2.24$ for $\Delta=1,\tfrac32,2$, and it tends to $2$ as $\Delta\to\infty$. We now make the saturation statement precise.
\begin{figure}[htbp]
    \centering
    \includegraphics[width=0.85\linewidth]{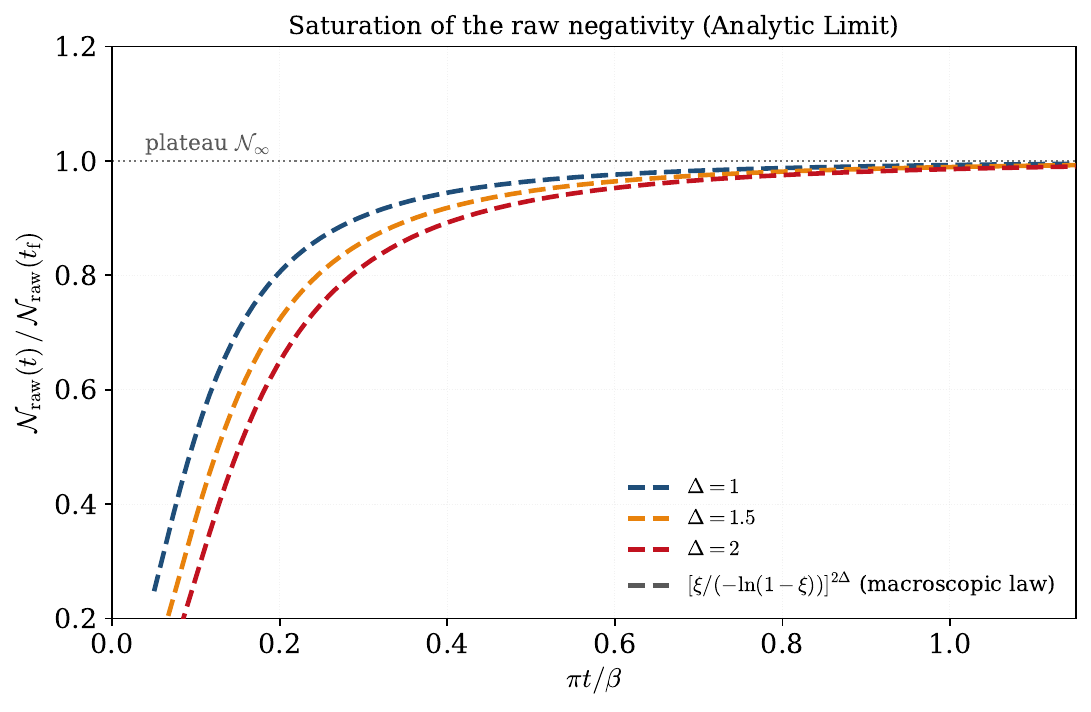}
    \caption{\textbf{Saturation of the raw negativity.}
Time evolution of the macroscopic law $[\xi/(-\ln(1-\xi))]^{2\Delta}$ (dashed curves, equation~\eqref{eq:raw_saturates}) for $\Delta\in\{1,3/2,2\}$ with $\varepsilon=0.1$ and $\beta=2\pi$. Each curve is normalized by its late time saturation value, so that the plateau sits at unity (dotted line). As $t\to\infty$ ($\xi\to0$) the geometric spreading exactly compensates the decay of the return amplitude, and the raw negativity saturates to the constant $\Neg_\infty$ of Proposition~\ref{prop:plateau}. The raw negativity is therefore bounded at late times, which is why the seed-normalized diagnostic is needed to follow the continued infall.}
    \label{fig:raw_saturation}
\end{figure}
\begin{proposition}[Saturation of the raw negativity]
\label{prop:plateau}
Within the macroscopic regime ($Q\gg1$, $\langle\NKrylov\rangle\ll D$), the raw Krylov--Wigner negativity approaches a constant plateau,
\beq
\Neg_{\mathrm{raw}}(t)\;\xrightarrow[t\to\infty]{}\;
\Neg_\infty\equiv\frac{c_\Delta\,\Gamma(2\Delta)}{2^{\Delta-2}\Gamma(\Delta)}
=\frac{2^{\Delta+1}}{\sqrt{\pi}}\,\Gamma\!\Bigl(\Delta+\tfrac12\Bigr)c_\Delta ,
\qquad
c_\Delta=\int_0^\infty\frac{|J_\Delta(Z)|}{Z^\Delta}\,\mathrm{d}Z ,
\label{eq:plateau_value}
\eeq
a finite number that depends only on the conformal dimension $\Delta$, and is independent of the temperature $\beta$, the regulator $\varepsilon$ and the Hilbert space dimension $D$. It is finite for every $\Delta>1/2$ and decreases monotonically with $\Delta$. The plateau is approached from below,
\beq
\Neg_{\mathrm{raw}}(t)=\Neg_\infty\Bigl(1-\Delta\,\xi(t)+O(\xi^2)\Bigr)
=\Neg_\infty\Bigl(1-\Delta\,\frac{\varepsilon_0}{\sinh^2(\pi t/\beta)}+\cdots\Bigr),
\label{eq:plateau_approach}
\eeq
i.e., exponentially on the thermal time scale, with relative deficit $\sim4\Delta\,\varepsilon_0\,e^{-2\pi t/\beta}$.
\end{proposition}
\noindent\textit{Proof.} Setting $u\equiv-\ln(1-\xi)=2|\ln|A||$, equation~\eqref{eq:raw_saturates} reads $\Neg_{\mathrm{raw}}=\Neg_\infty(\xi/u)^{2\Delta}$. Since $u=\xi+\tfrac12\xi^2+O(\xi^3)$, we have $\xi/u=1-\tfrac12\xi+O(\xi^2)$ and $(\xi/u)^{2\Delta}=1-\Delta\xi+O(\xi^2)$, which gives~\eqref{eq:plateau_approach} on inserting $\xi\simeq\varepsilon_0/\sinh^2(\pi t/\beta)$; the limit $\xi\to0$ yields~\eqref{eq:plateau_value}. The second form of $\Neg_\infty$ follows from Legendre duplication. Convergence of $c_\Delta$ is controlled at the two ends: near $Z=0$ the integrand tends to the constant $2^{-\Delta}/\Gamma(\Delta+1)$, while at large $Z$, $|J_\Delta(Z)|/Z^\Delta\sim\sqrt{2/\pi}\,|\cos(\cdots)|\,Z^{-\Delta-1/2}$, which is integrable precisely for $\Delta>1/2$. $\square$

\paragraph{Why a spreading distribution has bounded negativity.}
The saturation is not an accident of the integral but a balance between two competing effects. As $|A|\to1$ the mean Krylov level $\langle\NKrylov\rangle=2\Delta|A|^2/(1-|A|^2)\simeq2\Delta/\xi$ diverges, so the packet occupies an ever larger region of the chain; at the same time the peak height of $W$, set by the return probability $\Psurv=|\psi_0|^2=\xi^{2\Delta}$, vanishes. In the negativity these two effects are governed by the \emph{same} power of $\xi$: the radial spread contributes $\int_0^\infty Q^{2\Delta-1}e^{-2Q|\ln|A||}\,\mathrm{d}Q=\Gamma(2\Delta)\,(2|\ln|A||)^{-2\Delta}\propto\xi^{-2\Delta}$, which grows exactly as fast as the envelope $\Psurv\propto\xi^{2\Delta}$ decays. The total absolute quasi-probability is their product, and the two powers cancel. Physically, the quasi-probability spreads to fill an ever larger region of phase space while thinning in height, holding the $1$-norm of the signed distribution fixed once the state has delocalized. This is the discrete Wigner counterpart of a wavepacket whose non-classicality stops growing after it has spread over the phase space available to it.

\paragraph{Diagnostic status of $\Neg$, and what the matching means.}
Two points should be kept in mind, following the discussion at the end of section~\ref{sec:normalization_upfront}. (i) $\Neg$ is the raw negativity divided by the return probability; it is a rescaled diagnostic, not the Wigner negativity of a normalized state. (ii) Comparing~\eqref{eq:neg_full} with $1/\Psurv=(1-|A|^2)^{-2\Delta}$, the two share the same leading late time scaling $\sinh^{4\Delta}$ and differ only at subleading order ($2|\ln|A||=-\ln(1-\xi)=\xi+\tfrac12\xi^2+\cdots$ versus $1-|A|^2=\xi$). Thus the leading growth of $\Neg$ coincides with that of the inverse return probability, and the matching identified in section~\ref{sec:matching} is, at leading order, a statement about the survival probability decay exponent $2\Delta$ rather than about the magnitude of the magic. The magic content enters through the coefficient $\Neg_\infty$ of~\eqref{eq:plateau_value} and through the subleading structure. We therefore present the matching as a structural relation between boundary two-point data, and not as a dynamical coincidence of magic with tidal effects.
\begin{figure}[htbp]
    \centering
    \includegraphics[width=0.75\linewidth]{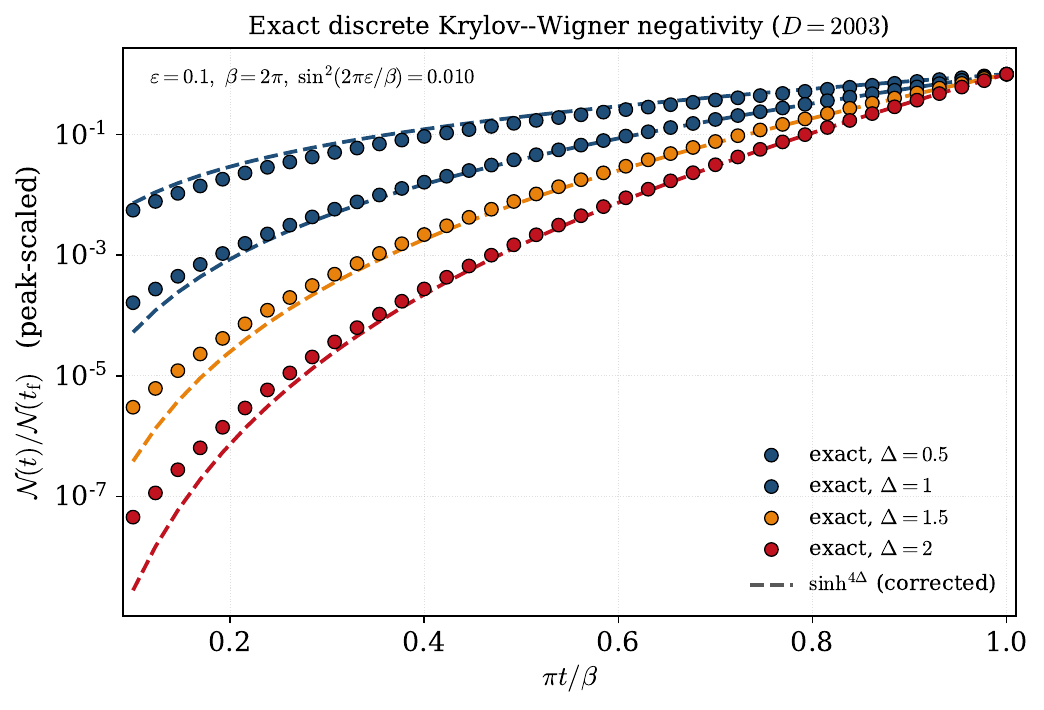}
    \caption{\textbf{Growth of the seed-normalized negativity (log--log).}
    The exact discrete seed-normalized negativity $\Neg=\Neg_{\mathrm{raw}}/\Psurv$ (markers), computed from the discrete Wigner function~\eqref{eq:wootters_def} on the $\mathbb{Z}_D\times\mathbb{Z}_D$ phase space of prime dimension $D=2003$ ($\varepsilon=0.1$, $\beta=2\pi$), plotted against $\sinh(\pi t/\beta)$ on log--log axes and divided by its value at $\pi t/\beta=1$. On these axes the continuum law $\Neg\sim\sinh^{4\Delta}$ of equation~\eqref{eq:neg_growth} is a straight line of slope $4\Delta$ (dashed); the exact markers follow it (fitted slopes $3.91,\,5.79,\,7.74$ for $\Delta=1,\tfrac32,2$, against $4,6,8$), confirming that the seed-normalized negativity \emph{grows} as a power law, in contrast to the raw negativity of figure~\ref{fig:raw_saturation}, which saturates. The window is kept within the validity regime of the macroscopic description ($\langle\NKrylov\rangle\ll D$, truncation weight $<10^{-3}$). The marginal case $\Delta=\tfrac12$ (open squares), where $c_\Delta$ diverges logarithmically, is shown for reference and carries a slow logarithmic enhancement above the pure $\sinh^{2}$ slope.}
    \label{fig:peak_scaled_negativity}
\end{figure}
\section{Matching and the Critical Dimension}
\label{sec:matching}
We now compare the normalized negativity rate~\eqref{eq:neg_rate_growth} with the growth rate of the variance of the Krylov wavepacket. Before doing so, we should say why this comparison is meaningful at all, since a priori the two are quantities of a different kind. The variance rate is a property of the occupation probabilities $|\psi_n|^2$ on the chain alone, while the negativity rate knows about the full phase space distribution, coherences included. What makes a clean comparison possible is the $\SU(1,1)$ coherent state structure: the entire evolved state is controlled by the single function $|A(t)|^2$, so \emph{every} observable -- moments and negativity alike -- is a function of the one scaling variable $\xi(t)$ of~\eqref{eq:xi_def}, and any two observables are locked into a definite functional relation. The comparison of the two rates then asks a sharp question: does the production of phase space non-classicality keep pace with the dispersion of the wavepacket? Within the coherent state family the answer is encoded in two power laws of $\xi$, and the comparison reduces to matching exponents.

We should stress that this locking is special to the one-parameter $\SU(1,1)$ family, and is \emph{not} expected for generic dynamics, which explores the Hilbert space well beyond a single coherent state orbit. For instance, in Gaussian random matrix theory the Krylov--Wigner negativity grows as a power law for a long time and then saturates close to its upper bound of $\sqrt{D}$~\cite{Basu:2024, Basu:2025rm}, and there is no reason for it to track the variance of the Krylov position. The matching below is therefore a statement about the semiclassical, symmetry-dominated window of the quench -- which is precisely the window in which the bulk dictionary~\eqref{eq:proper_momentum} is formulated.

\paragraph{Exact variance rate.}
From~\eqref{eq:nb_var_lemma}, $\mathrm{Var}(\NKrylov)=2\Delta(1-\xi)/\xi^2$, whose exact time derivative is
\beq
\boxed{\;
\frac{\mathrm{d}}{\mathrm{d}t}\mathrm{Var}(\NKrylov)
=\frac{4\Delta\pi}{\beta}\cdot
\frac{\sinh(\pi t/\beta)\cosh(\pi t/\beta)\bigl[2\sinh^2(\pi t/\beta)+\varepsilon_0\bigr]}
{\varepsilon_0^2}\; . }
\label{eq:var_rate_exact}
\eeq
This is exact for all $t>0$ and $\varepsilon>0$. For $S^2\gg\varepsilon_0$ the bracket reduces to $2S^2$, giving the late time asymptotic
\beq
\frac{\mathrm{d}}{\mathrm{d}t}\mathrm{Var}(\NKrylov)
\;\approx\;\frac{8\Delta\pi}{\beta}\,
\frac{\sinh^3(\pi t/\beta)\cosh(\pi t/\beta)}{\sin^4(2\pi\varepsilon/\beta)} .
\label{eq:var_rate_late}
\eeq
\paragraph{Matching through the single scaling variable.}
Both late time rates~\eqref{eq:neg_rate_growth} and~\eqref{eq:var_rate_late} are powers of $\xi\simeq\varepsilon_0/S^2$, so their $t$- and $\varepsilon$-exponents are locked together:
\begin{center}
\renewcommand{\arraystretch}{1.6}
\begin{tabular}{lccc}
\toprule
Quantity & $t$-scaling & $\varepsilon$-scaling & Match condition\\
\midrule
$\NegRate(t)$
 & $\sinh^{4\Delta-1}(\pi t/\beta)\cosh(\pi t/\beta)$
 & $\sin^{-4\Delta}(2\pi\varepsilon/\beta)$ & \\
$\dfrac{\mathrm{d}\,\mathrm{Var}}{\mathrm{d}t}$
 & $\sinh^{3}(\pi t/\beta)\cosh(\pi t/\beta)$
 & $\sin^{-4}(2\pi\varepsilon/\beta)$ & \\[4pt]
\midrule
\textit{Equality requires} & $4\Delta-1=3$ & $4\Delta=4$ & $\Delta=1$\\
\bottomrule
\end{tabular}
\end{center}
\noindent Both columns give the same condition,
\beq
4\Delta-1=3\;\Longrightarrow\;\Delta=1 ,
\label{eq:matching_correct}
\eeq
and they are not independent tests, since the $t$- and $\varepsilon$-exponents are tied together by $\xi$. At this value the two rates share the functional form $\sinh^3(\pi t/\beta)\cosh(\pi t/\beta)/\sin^4(2\pi\varepsilon/\beta)$, and we obtain the late time proportionality
\beq
\boxed{\;
\NegRate(t)\;\simeq\;\mathcal{A}(\beta,\varepsilon)\,
\frac{\mathrm{d}}{\mathrm{d}t}\mathrm{Var}(\NKrylov) ,
\qquad \Delta=1,\quad \sinh^2(\pi t/\beta)\gg\sin^2(2\pi\varepsilon/\beta) , }
\label{eq:rate_equality}
\eeq
with $\mathcal{A}(\beta,\varepsilon)$ a $t$-independent constant. Away from $\Delta=1$ the ratio $\NegRate/(\mathrm{d}\,\mathrm{Var}/\mathrm{d}t)\propto\sinh^{4\Delta-4}(\pi t/\beta)$ is not constant, so the proportionality holds if and only if $\Delta=1$.

We should be clear that~\eqref{eq:rate_equality} is an asymptotic statement, not an exact one. The variance rate~\eqref{eq:var_rate_exact} is exact and retains the factor $2\sinh^2(\pi t/\beta)+\varepsilon_0$, whereas the negativity rate comes from the macroscopic Bessel form; the ratio of the two reaches its constant plateau only once $\sinh^2(\pi t/\beta)\gg\varepsilon_0$, and is not constant at finite time. No exact finite time equality is claimed.

\paragraph{The general negativity--variance relation and the critical dimension.}
The matching admits a compact statement valid for \emph{every} $\Delta>\tfrac12$. At late times $\Neg\propto\xi^{-2\Delta}$ by~\eqref{eq:neg_growth}, while $\mathrm{Var}(\NKrylov)=2\Delta(1-\xi)/\xi^{2}\to2\Delta/\xi^{2}$, so eliminating the common scaling variable $\xi$ gives
\beq
\boxed{\;\Neg(t)\;\propto\;\bigl[\mathrm{Var}(\NKrylov)\bigr]^{\Delta}\;,
\qquad \sinh^2(\pi t/\beta)\gg\varepsilon_0,\quad \Delta>\tfrac12\;:}
\label{eq:general_relation}
\eeq
the normalized negativity is a fixed power of the second moment of the Krylov wavepacket at any dimension. The value $\Delta=1$ is then distinguished as the \emph{unique} dimension at which this relation linearizes, $\Neg\propto\mathrm{Var}(\NKrylov)$. Only there do the two rates coincide, as in~\eqref{eq:rate_equality}, and only there does the negativity rate acquire the clean geometric reading of section~\ref{sec:tidal_momentum}, $\NegRate\propto\Cplx\,P_\rho$ -- the rate of tidal stretching of neighboring infalling geodesics. Away from $\Delta=1$ the negativity rate is the tidal rate dressed by the additional factor $\Cplx^{2\Delta-2}$.\footnote{In bulk language, $\Delta$ throughout this paper is the holomorphic weight $h$ of the quenching primary; a spinless bulk scalar carries total dimension $\Delta_{\mathrm{tot}}=2\Delta$, so $\Delta=1$ corresponds to $m^2\lAdS^2=\Delta_{\mathrm{tot}}(\Delta_{\mathrm{tot}}-2)=0$, a massless scalar. (The Breitenlohner--Freedman bound~\cite{Breitenlohner:1982jf} is instead saturated at $\Delta=\tfrac12$ in this convention.) We attach no bulk stability interpretation to $\Delta=1$; its significance is precisely the linearization of~\eqref{eq:general_relation}.} That the higher moment (negativity) diagnostic singles out this linearization point, while the first moment (complexity) diagnostic is insensitive to it, is the structural observation of this work.

Figures~\ref{fig:var_growth}--\ref{fig:ratio_delta} illustrate~\eqref{eq:rate_equality} using the exact formulas~\eqref{eq:nb_var_lemma} and~\eqref{eq:var_rate_exact}; the diagnostic ratio scales as $\sinh^{4\Delta-4}(\pi t/\beta)$ and is asymptotically flat only at $\Delta=1$.
\begin{figure}[t]
\centering
\includegraphics[width=0.6\linewidth]{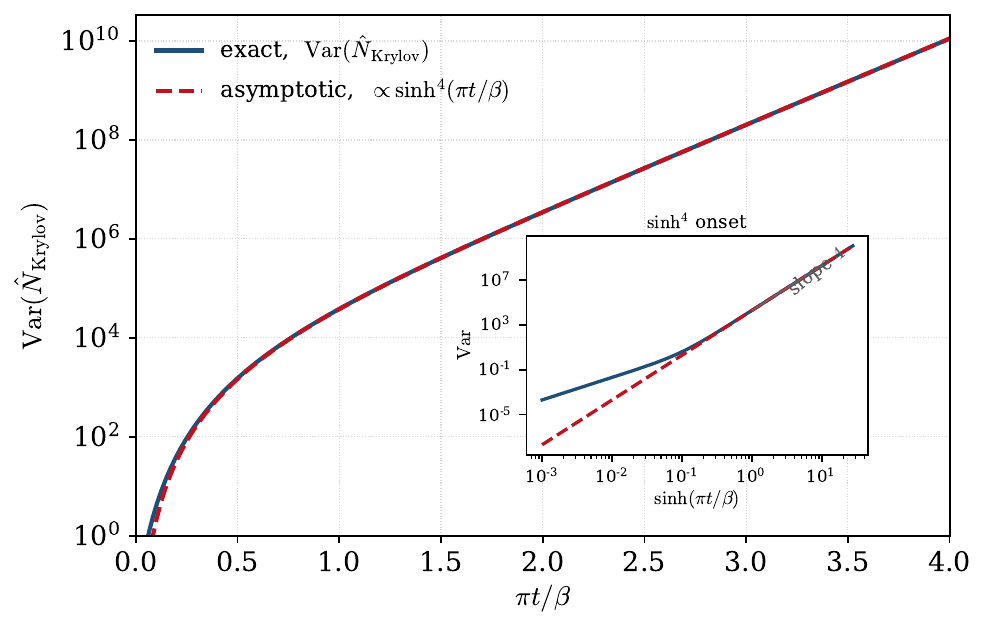}
\caption{\textbf{Exact growth of the Krylov wavepacket variance.}
$\mathrm{Var}(\NKrylov)$ versus $\pi t/\beta$ from the negative binomial formula~\eqref{eq:nb_var_lemma} (solid) and its late time $\sinh^4$ asymptotic (dashed), for $\sin^2(2\pi\varepsilon/\beta)=0.01$ at $\Delta=1$.}
\label{fig:var_growth}
\end{figure}
\begin{figure}[htbp]
\centering
\includegraphics[width=0.85\linewidth]{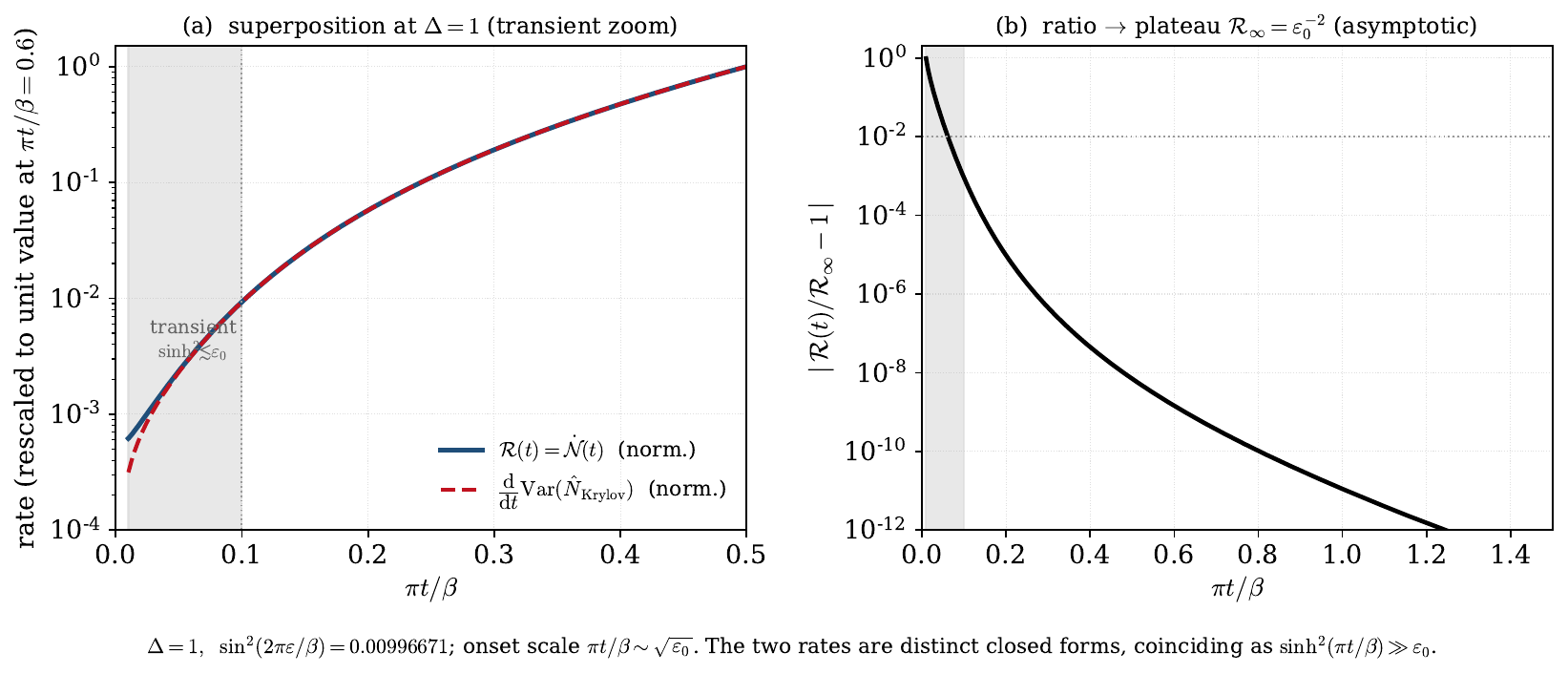}
\caption{\textbf{Superposition of $\NegRate(t)$ and $\mathrm{d}\,\mathrm{Var}/\mathrm{d}t$ at $\Delta=1$.} \emph{(a)} The normalized negativity rate $\NegRate=\dot{\Neg}$ and the exact Krylov variance rate $\tfrac{\mathrm{d}}{\mathrm{d}t}\mathrm{Var}(\NKrylov)$, each rescaled to unit value at the right edge, versus $\pi t/\beta$ (log scale), zoomed to the transient region for $\sin^2(2\pi\varepsilon/\beta)=0.01$. The grey band marks the transient $\sinh^2(\pi t/\beta)\lesssim\varepsilon_0$, i.e., $\pi t/\beta\lesssim\sqrt{\varepsilon_0}$, where the two rates differ; they coincide thereafter. The curves are built from \emph{different} closed forms -- $\NegRate$ from the macroscopic negativity $\Neg\propto[\ln(1+\varepsilon_0/\sinh^2(\pi t/\beta))]^{-2\Delta}$ and $\tfrac{\mathrm{d}}{\mathrm{d}t}\mathrm{Var}$ from the exact rate~\eqref{eq:var_rate_exact} -- so their agreement is not a tautology. \emph{(b)} Relative deviation of the ratio $\NegRate/\tfrac{\mathrm{d}}{\mathrm{d}t}\mathrm{Var}$ from its late time plateau $\NegRate_\infty=\varepsilon_0^{-2}$ (log scale): it decreases monotonically to zero as $t\to\infty$, so the proportionality becomes exact only in the asymptotic regime $\sinh^2(\pi t/\beta)\gg\varepsilon_0$.}
\label{fig:rate_comparison}
\end{figure}
\begin{figure}[htbp]
\centering
\includegraphics[width=0.55\linewidth]{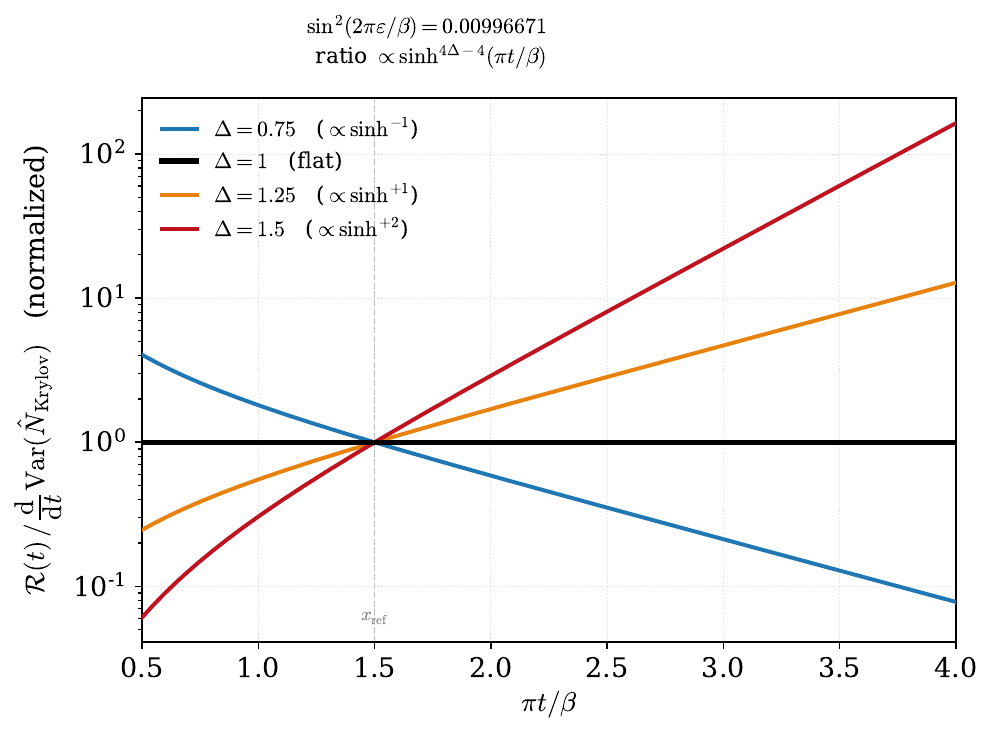}
\caption{\textbf{Diagnostic of the matching condition, with the asymptotics shown explicitly.}
The ratio $\NegRate(t)\big/\tfrac{\mathrm{d}}{\mathrm{d}t}\mathrm{Var}(\NKrylov)$ for several $\Delta$ at $\sin^2(2\pi\varepsilon/\beta)=0.01$, each normalized to unity at $x_{\mathrm{ref}}\equiv\pi t/\beta=1.5$ (so that every $\Delta$-dependent prefactor, including $c_\Delta$ and the $O(1)$ lattice-convention factor, cancels; the plot is convention independent). \emph{(a)} The exact ratio (solid), built from the macroscopic negativity $\Neg\propto[\ln(1+\varepsilon_0/\sinh^2(\pi t/\beta))]^{-2\Delta}$ and the exact variance rate~\eqref{eq:var_rate_exact}, peels away from its late time asymptote $(\sinh(\pi t/\beta)/\sinh x_{\mathrm{ref}})^{4\Delta-4}$ (dashed) at early times and merges onto it at late times. \emph{(b)} The local logarithmic slope $\mathrm{d}\ln\NegRate/\mathrm{d}\ln\sinh(\pi t/\beta)$ converges to the constant $4\Delta-4$ (dotted guides). The ratio is thus flat \emph{only} at the critical value $\Delta=1$, where the negativity--variance relation~\eqref{eq:general_relation} linearizes; it decays for $\Delta<1$ and diverges for $\Delta>1$. The flatness is an asymptotic statement, as discussed below equation~\eqref{eq:rate_equality}.}
\label{fig:ratio_delta}
\end{figure}
\newpage
\section{Negativity Rate as Tidal Momentum}
\label{sec:tidal_momentum}
The matching acquires a sharper bulk reading in the language of the holographic Krylov dictionary, in which the rate of spread complexity is a proper momentum. We show that the normalized negativity rate is governed by the \emph{same} proper momentum, dressed by an explicit factor, and that this recasts the $\Delta=1$ condition as the equality of two bulk momenta. The relations of this section rest on the exact boundary identity~\eqref{eq:var_mean_identity} and the dictionary~\eqref{eq:proper_momentum}; they do not invoke the speculative string considerations of section~\ref{sec:discussion}.

\paragraph{Bulk setup.}
Since we will use the bulk picture repeatedly, let us first set it up explicitly. The bulk geometry is the planar (one-sided) BTZ black brane~\cite{Banados:1992wn},
\beq
\mathrm{d}s^2=-f(r)\,\mathrm{d}t^2+\frac{\mathrm{d}r^2}{f(r)}+\frac{r^2}{\lAdS^2}\,\mathrm{d}x^2 ,
\qquad
f(r)=\frac{r^2-r_h^2}{\lAdS^2} ,
\qquad
r_h=\frac{2\pi\lAdS^2}{\beta} ,
\label{eq:btz_metric}
\eeq
whose boundary time $t$ coincides with the CFT time of the quench~\eqref{eq:local_quench}, and whose Hawking temperature is $1/\beta$. It is convenient to trade $r$ for the proper radial distance $\rho$ from the horizon, defined by $r=r_h\cosh(\rho/\lAdS)$, in terms of which $\mathrm{d}r^2/f=\mathrm{d}\rho^2$. The local operator insertion creates a massive particle near the cutoff boundary at $t=0$, at a radial depth set by the smearing $\varepsilon$; it subsequently falls along a radial timelike geodesic $r(\tau)$ with conserved energy $E=f\,\mathrm{d}t/\mathrm{d}\tau$, where $\tau$ is its proper time. The proper radial momentum $P_\rho$ in the dictionary below is the radial momentum of this geodesic as measured by the static observer at the particle's instantaneous position, expressed as a function of the boundary time $t$ at which that position is reached~\cite{Caputa:2024}.

\paragraph{Boundary time versus proper time.}
One more clarification is needed before we proceed. All the rates in this paper -- $\dot\Cplx$, $\NegRate$ and $\mathrm{d}\,\mathrm{Var}/\mathrm{d}t$ -- are taken with respect to the boundary time $t$, while the geodesic deviation equation~\eqref{eq:geodesic_deviation_body} below is written, as usual, in the proper time $\tau$ of the central geodesic. The two are related along the trajectory by the redshift $\mathrm{d}t/\mathrm{d}\tau=E/f(r(\tau))$. Converting a first $\tau$-derivative into a $t$-derivative therefore costs one factor of $\mathrm{d}\tau/\mathrm{d}t$. But the tidal reading below only compares the boundary time rate $\NegRate$ with the boundary time combination $\Cplx\,P_\rho$, both referred to the same static frame at the position of the particle and parametrized by $t$; the common redshift factor affects the two sides of the comparison identically, and drops out of the identification~\eqref{eq:tidal_duality}. It would of course reappear if one wanted to compare with a bulk rate defined in proper time.

Recall the established first moment entry~\cite{Caputa:2024}: with $\mathcal{C}(t)\equiv\langle\NKrylov\rangle$,
\beq
\frac{d\mathcal{C}}{dt}\;\propto\;P_\rho(t) ,
\label{eq:proper_momentum_recall}
\eeq
and $\mathcal{C}(t)=\tfrac{2\Delta}{\varepsilon_0}\sinh^2(\pi t/\beta)$ exactly, so $P_\rho\propto\sinh(\pi t/\beta)\cosh(\pi t/\beta)$.

\paragraph{Variance rate (exact).}
The Krylov occupation is negative binomial, so the variance is an exact quadratic function of the mean, $\mathrm{Var}=\mathcal{C}+\mathcal{C}^2/2\Delta$ by~\eqref{eq:var_mean_identity}. Differentiating,
\beq
\frac{d\,\mathrm{Var}(\NKrylov)}{dt}
= \left(1+\frac{\mathcal{C}}{\Delta}\right)\frac{d\mathcal{C}}{dt}
\;\propto\;\left(1+\frac{\mathcal{C}(t)}{\Delta}\right)P_\rho(t) .
\label{eq:var_rate_momentum}
\eeq
This is an \emph{exact} boundary identity: the variance rate is the proper radial momentum dressed by the calculable factor $1+\mathcal{C}/\Delta$.

\paragraph{Normalized negativity rate.}
From the macroscopic result $\Neg\propto\sinh^{4\Delta}(\pi t/\beta)$ of section~\ref{sec:growth}, $\NegRate\propto\sinh^{4\Delta-1}\cosh$. Using $\mathcal{C}\propto\sinh^2$ and $P_\rho\propto\sinh\cosh$,
\beq
\NegRate(t)\;\propto\;\sinh^{4\Delta-2}\cdot(\sinh\cosh)
\;\propto\;\mathcal{C}(t)^{\,2\Delta-1}\,P_\rho(t) .
\label{eq:neg_rate_momentum}
\eeq
Thus the negativity rate is the proper momentum dressed by $\mathcal{C}^{2\Delta-1}$, whereas the variance rate~\eqref{eq:var_rate_momentum} is dressed by $1+\mathcal{C}/\Delta\to\mathcal{C}/\Delta$ at late times. The two dressings coincide precisely when $2\Delta-1=1$, i.e., $\Delta=1$, which reproduces~\eqref{eq:matching_correct} from the bulk momentum side. At that dimension,
\beq
\boxed{\;
\NegRate(t)\;\propto\;\mathcal{C}(t)\,P_\rho(t)
\qquad(\Delta=1)\; ,}
\label{eq:tidal_duality}
\eeq
the negativity rate is the product of the proper radial position (encoded in $\mathcal{C}$) and the proper radial momentum.

\paragraph{Geometric interpretation: geodesic deviation.}
The structure of~\eqref{eq:tidal_duality} suggests which bulk quantity the negativity rate is tracking. The variance is the squared radial spread of the wavepacket, $\mathrm{Var}\leftrightarrow(\delta\rho)^2$, whose rate is $\tfrac{d}{dt}(\delta\rho)^2=2\,\delta\rho\,\dot{(\delta\rho)}$. The separation $\delta\rho$ between neighboring infalling geodesics obeys the geodesic deviation equation
\beq
\frac{D^2\,\delta\rho}{d\tau^2}=-R^{\rho}_{\ \tau\tau\rho}\,\delta\rho ,
\label{eq:geodesic_deviation_body}
\eeq
written in the proper time $\tau$ of the central geodesic; re-expressed in boundary time via the redshift factor discussed above, the first derivatives on both sides of the comparison acquire the same factor of $\mathrm{d}\tau/\mathrm{d}t$, which cancels. The relative (tidal) momentum $\dot{(\delta\rho)}$ grows in proportion to the separation times the infall rate, and $\delta\rho\,\dot{(\delta\rho)}\propto\mathcal{C}\,P_\rho$ is the rate of tidal stretching. In this reading, the complexity rate measures the classical geodesic momentum, while the normalized negativity rate measures the tidal momentum of geodesic deviation, i.e., the quantum spreading about the classical trajectory. We regard this as a heuristic interpretation rather than a derivation: it is expressed through the proper momentum dictionary alone, is independent of the string considerations of section~\ref{sec:discussion}, and is subject to the caveat on the diagnostic status of $\Neg$ discussed at the end of section~\ref{sec:growth}.
\section{Discussion}
\label{sec:discussion}
We have presented a boundary calculation -- the seed-normalized Krylov--Wigner distribution, its total negativity, and the negativity rate -- and its interpretation as a second moment companion to spread complexity. Let us summarize what has been established, and then say a few words, clearly flagged as speculative, about a possible bulk carrier of the negativity.

\subsection*{Summary of boundary results}
\textit{Seed-normalized calculus.}
By fixing the normalization at the level of the amplitudes -- dividing out the return amplitude to form $\chi_n=\psi_n/\psi_0$ and the seed-normalized distribution $\Wnorm=W/\Psurv$ (section~\ref{sec:seed_normalized}) -- the physically informative object is primary from the start. Its closed Bessel form~\eqref{eq:wigner_bessel} follows from the $\SU(1,1)$ coherent state structure, Stirling asymptotics and the Euler--Maclaurin formula, and carries no survival factor.

\textit{Negativity and its meaning.}
Evaluating the total negativity directly from the Bessel form, and keeping the Bessel variable $Z=Q\tilde\theta$ throughout, the radial power is $Q^{2\Delta-1}$ and the normalized negativity grows as $\Neg(t)\sim\sinh^{4\Delta}(\pi t/\beta)$, equation~\eqref{eq:neg_growth}. The raw, normalized-state negativity instead saturates, equation~\eqref{eq:raw_saturates}, consistent with $\Neg_{\mathrm{raw}}\le\sqrt{D}$ and with the $O(1)$ plateau found in~\cite{Balasubramanian:2026wn}. As stressed in sections~\ref{sec:normalization_upfront} and~\ref{sec:growth}, $\Neg$ is a rescaled (survival-normalized) diagnostic whose leading growth coincides with the inverse return probability; the matching is correspondingly a statement about boundary two-point data. Proposition~\ref{prop:neg_rate} (the rate, including the descendant weight term) and the exact variance rate~\eqref{eq:var_rate_exact} are exact boundary results.

\textit{Matching and tidal momentum.}
The normalized negativity is at late times a fixed power of the Krylov variance, $\Neg\propto[\mathrm{Var}(\NKrylov)]^{\Delta}$, equation~\eqref{eq:general_relation}, for every $\Delta>\tfrac12$, and the relation linearizes precisely at $\Delta=1$: only there does the normalized negativity rate track the Krylov variance rate~\eqref{eq:rate_equality}, a necessary, leading order, asymptotic condition. In the proper momentum language of~\cite{Caputa:2024} the rate is $\NegRate\propto\Cplx^{2\Delta-1}P_\rho$, coinciding with the variance dressing $(1+\Cplx/\Delta)P_\rho$ at $\Delta=1$, where $\NegRate\propto\Cplx\,P_\rho$, which we interpret as the rate of tidal stretching of neighboring infalling geodesics (section~\ref{sec:tidal_momentum}).

\subsection*{A speculative bulk candidate for the negativity}
We end with a remark which is strictly a direction for future work, and logically independent of the boundary results above. It is natural to wonder what bulk object could source the Wigner negativity in the first place. A structureless point particle cannot: its phase space distribution is Gaussian and positive, and remains so under the quadratic evolution about a geodesic, generating no negativity. So any bulk carrier of the negativity must possess genuine quantum internal degrees of freedom, and a natural candidate in $\AdS_3$ is the quantized fundamental string. String propagation on $\AdS_3$ is described by the $\SL(2,\mathbb{R})$ WZW model~\cite{Maldacena:2000hw, Maldacena:2001km}, whose discrete series sector (with $j=\Delta$ under the standard dictionary~\cite{Balasubramanian:1998sn, Horowitz:1999jd}) realizes the same representation $D^+_\Delta$ of $\SU(1,1)$ -- and hence the same quadratic Casimir $\Delta(\Delta-1)$ -- as the boundary Krylov chain of section~\ref{sec:su11_krylov}. The operator measuring the transverse size of the string~\cite{Susskind:1993aa} is then a candidate bulk counterpart of the Krylov position operator $\NKrylov$, which would map the Krylov variance to the fluctuations of the transverse string size, and the negativity rate to the growth of the squared transverse string area. Of course, a shared representation is a necessary condition for such an identification, but by no means a sufficient one -- in particular, the transverse size operator does not obviously share the integer spectrum of $\NKrylov$ -- and the mechanism populating the transverse modes during infall would have to be formulated in the freely falling frame, where no local temperature or curvature divergence is available at the smooth BTZ horizon. Developing, or falsifying, this identification is left to future work; for a recent treatment of extended probes of the proper momentum dictionary, see~\cite{Fatemiabhari:2025}.

\subsection*{Outlook}
There are several boundary questions that remain open:
\begin{enumerate}
\item The operational status of the seed-normalized negativity $\Neg$ deserves sharpening -- whether the physically preferred observable is the raw (bounded, saturating) negativity, the survival-normalized $\Neg$, or a differently conditioned quantity -- since this fixes the physical content of the matching.
\item The proportionality constant $\mathcal{A}(\beta,\varepsilon)$ in~\eqref{eq:rate_equality} should be determined by a microstate-level analysis, rather than fixed by the macroscopic coefficient alone.
\item The same phase space calculus can be brought to bear on the stabilizer complexity of Hawking radiation~\cite{Basu:2025, Page:1993wv, Penington:2019npb}, on the Krylov complexity of purification for thermal states~\cite{Das:2024}, and on toy models such as DSSYK, where the Krylov basis coincides with the chord-length basis~\cite{Basu:2024, Balasubramanian:2022dnj, Rabinovici:2023}.
\end{enumerate}
\newpage
\begin{acknowledgments}
I am deeply grateful to Prof.~Onkar Parrikar for his mentorship, comprehensive assistance, and support
throughout this project. I want to thank Prof.~Saskia Demulder and my friend Dr.~Rathindranath for a careful and extensive review of an earlier draft,
whose detailed comments---on the phase-space negativity integral, the interpretation of the normalized
diagnostic, the asymptotic status of the matching, the near-horizon frame, and the transverse-mode
weighting---substantially improved both the derivations and their presentation. I also thank
Avijit Sinha, Pruthvi Suryadevara, and Jyotirmoy Mukherjee for helpful discussions, and acknowledge
the support of the Department of Theoretical Physics, TIFR, Mumbai. I thank the anonymous referee for a careful reading and for a number of comments that improved the presentation.
\end{acknowledgments}
\noindent\textbf{Note added.}
While this manuscript was being completed, the work of Balasubramanian, Caputa, Parrikar and Singh~\cite{Balasubramanian:2026wn} appeared, in which the Wigner negativity in Krylov space is studied as a probe of emergent semiclassicality across several solvable models. For 2d CFT local quenches those authors find that the state Wigner negativity remains an $O(1)$ constant at late times; this is precisely the raw negativity plateau of Proposition~\ref{prop:plateau}, and the two computations agree where they overlap. The seed-normalized diagnostic studied in the present paper isolates the growing second moment information carried on top of that plateau.
\newpage
\appendix
\section{Elementary verifications of the \texorpdfstring{$\su(1,1)$}{su(1,1)} structure}
\label{app:su11}
For completeness we record here the elementary checks of the relations~\eqref{eq:su11_comm}--\eqref{eq:casimir_K} of section~\ref{sec:su11_krylov}; none of this material is new~\cite{Caputa:2021sib, Caputa:2021ori, Balasubramanian:2022tpr, Perelomov:1986}.
\paragraph{(1) Commutation relations.}
Acting on a generic basis vector and using \eqref{eq:K0K}--\eqref{eq:KmK},
\begin{align}
K_0^{(K)}K_+^{(K)}|K_n\rangle &= (n+1+\Delta)\sqrt{(n+1)(n+2\Delta)}\,|K_{n+1}\rangle , \\
K_+^{(K)}K_0^{(K)}|K_n\rangle &= (n+\Delta)\sqrt{(n+1)(n+2\Delta)}\,|K_{n+1}\rangle ,
\end{align}
whose difference is $[K_0^{(K)},K_+^{(K)}]|K_n\rangle = K_+^{(K)}|K_n\rangle$, and likewise $[K_0^{(K)},K_-^{(K)}]=-K_-^{(K)}$ by Hermitian conjugation. For the remaining relation,
\begin{align}
K_-^{(K)}K_+^{(K)}|K_n\rangle &= (n+1)(n+2\Delta)|K_n\rangle , \\
K_+^{(K)}K_-^{(K)}|K_n\rangle &= n(n+2\Delta-1)|K_n\rangle ,
\end{align}
so that $[K_-^{(K)},K_+^{(K)}]|K_n\rangle = (2n+2\Delta)|K_n\rangle = 2K_0^{(K)}|K_n\rangle$, which is~\eqref{eq:su11_comm}.
\paragraph{(2) Discrete series.}
The coefficient in~\eqref{eq:KmK} vanishes at $n=0$, so the lowering operator annihilates the seed, $K_-^{(K)}|K_0\rangle=0$: the chain has a bottom rung and no states of negative level, which is the defining property of a lowest weight module. From~\eqref{eq:K0K} the spectrum of $K_0^{(K)}$ is $\{\Delta,\Delta+1,\Delta+2,\dots\}$, equally spaced and bounded below by $\Delta$, which is exactly the weight content of the positive discrete series $D_\Delta^+$ with Bargmann index $k=\Delta$~\cite{Perelomov:1986}. The index is therefore not a free label; it is pinned to the conformal dimension of the quenching operator.
\paragraph{(3) Casimir.}
Using $K_+^{(K)}K_-^{(K)}|K_n\rangle = n(n+2\Delta-1)|K_n\rangle$ from above,
\begin{align}
C_2^{(K)}|K_n\rangle
&= \bigl[(n+\Delta)(n+\Delta-1) - n(n+2\Delta-1)\bigr]|K_n\rangle \notag \\
&= \bigl[\,n^2+n(2\Delta-1)+\Delta(\Delta-1) - n^2 - n(2\Delta-1)\,\bigr]|K_n\rangle
= \Delta(\Delta-1)|K_n\rangle .
\end{align}
The level dependent pieces $n^2$ and $n(2\Delta-1)$ cancel identically, so $C_2^{(K)}$ takes the same value $\Delta(\Delta-1)$ on every rung of the chain, which is the operator identity~\eqref{eq:casimir_K}. By Schur's lemma this constancy is the statement that the entire chain lies in a single irreducible representation.
\paragraph{(4) Lanczos data.}
The matrix elements $\sqrt{n(n+2\Delta-1)}$ of the lowering operator in~\eqref{eq:KmK} differ from the Lanczos coefficients~\eqref{eq:lanczos_CFT} only by the overall constant $\pi/[\beta\sin(2\pi\varepsilon/\beta)]$; this is the value of $\alpha$ quoted below~\eqref{eq:lanczos_sl2}, and matches the microscopic data exactly.
\section{Stirling asymptotics and the Gamma kernel}
\label{app:stirling}
The continuum limit rests on the large-$Q$ behavior of the kernel in~\eqref{eq:Wnorm_discrete},
\beq
K(Q,m) = \frac{\sqrt{\Gamma(2\Delta+Q+m)\,\Gamma(2\Delta+Q-m)}}{\sqrt{(Q+m)!\,(Q-m)!}} .
\label{eq:kernel_def}
\eeq
Writing $(Q\pm m)!=\Gamma(Q\pm m+1)$, the square factorizes into two ratios of identical structure,
\beq
K(Q,m)^2
= \frac{\Gamma(2\Delta+Q+m)}{\Gamma(Q+m+1)}\cdot
  \frac{\Gamma(2\Delta+Q-m)}{\Gamma(Q-m+1)} .
\label{eq:kernel_squared}
\eeq
Each factor is $\Gamma(z+a)/\Gamma(z+b)$ with $z=Q\pm m$, $a=2\Delta$, $b=1$. Stirling's expansion $\ln\Gamma(z)=(z-\tfrac12)\ln z-z+\tfrac12\ln(2\pi)+O(z^{-1})$ gives
\beq
\frac{\Gamma(z+a)}{\Gamma(z+b)}
= z^{\,a-b}\left[\,1 + \frac{(a-b)(a+b-1)}{2z} + O\!\left(z^{-2}\right)\right] ,
\label{eq:gamma_ratio}
\eeq
so with $a-b=2\Delta-1$,
\beq
K(Q,m)^2
= (Q+m)^{2\Delta-1}(Q-m)^{2\Delta-1}
  \left[\,1 + \Delta(2\Delta-1)\!\left(\tfrac{1}{Q+m}+\tfrac{1}{Q-m}\right) + \cdots\right] ,
\label{eq:kernel_expanded}
\eeq
and taking the positive square root and passing to $x=m/Q\in[-1,1]$,
\beq
\boxed{\;K(Q,x) \simeq Q^{2\Delta-1}\bigl(1-x^2\bigr)^{\Delta-1/2}\;.}
\label{eq:stirling_result}
\eeq
The relative correction is uniformly $O(1/Q)$ for $x$ bounded away from $\pm1$. Near $x=\pm1$ Stirling's approximation degrades, but the weight $(1-x^2)^{\Delta-1/2}$ vanishes there for every $\Delta>1/2$, so the boundary region carries negligible measure -- the same vanishing that removes the Euler--Maclaurin boundary terms in appendix~\ref{app:euler}.
\section{Euler--Maclaurin continuum limit}
\label{app:euler}
With appendix~\ref{app:stirling}, the mode sum in~\eqref{eq:Wnorm_discrete} becomes
\beq
\sum_{m=-Q}^{Q} f(m) , \qquad
f(m) = K(Q,m)\,e^{im\theta}
\;\simeq\; Q^{2\Delta-1}\Bigl(1-\tfrac{m^2}{Q^2}\Bigr)^{\Delta-1/2} e^{iZm/Q} ,
\label{eq:em_summand}
\eeq
with $\theta\equiv 4\pi p/D+2\varphi_A$ and $Z\equiv Q\theta$. The Euler--Maclaurin formula
\beq
\sum_{m=-Q}^{Q} f(m)
= \int_{-Q}^{Q} f\,\mathrm{d}m
+ \frac{f(Q)+f(-Q)}{2}
+ \sum_{k\ge1}\frac{B_{2k}}{(2k)!}\Bigl[f^{(2k-1)}(Q)-f^{(2k-1)}(-Q)\Bigr] + R
\label{eq:euler_maclaurin}
\eeq
has only boundary corrections. Changing variables to $x=m/Q$,
\beq
\sum_{m=-Q}^{Q}f(m)\approx Q^{2\Delta}\int_{-1}^{1}(1-x^2)^{\Delta-1/2}e^{iZx}\,\mathrm{d}x .
\label{eq:continuum_limit}
\eeq
For $\Delta>1/2$ the weight has an algebraic zero at $x=\pm1$, so $f(\pm Q)=0$; the residual mismatch between the sum and the integral is suppressed by a positive power of $Q$ (in the boundary layer $m=Q-s$, $f\simeq2^{\Delta-1/2}Q^{\Delta-1/2}s^{\Delta-1/2}e^{iZm/Q}$, of order $Q^{\Delta-1/2}$ against the leading $Q^{2\Delta}$, giving a relative $O(Q^{-(\Delta+1/2)})$; the smooth interior contributes $O(Q^{-2})$). The integral~\eqref{eq:continuum_limit} is Poisson's Bessel representation~\eqref{eq:poisson_bessel}; reinstating the prefactor $|A|^{2Q}/(D\,\Gamma(2\Delta))$ and collapsing via Legendre duplication $\Gamma(2\Delta)=2^{2\Delta-1}\pi^{-1/2}\Gamma(\Delta)\Gamma(\Delta+\tfrac12)$,
\beq
\frac{1}{\Gamma(2\Delta)}\cdot\sqrt{\pi}\,\Gamma\!\left(\Delta+\tfrac12\right)\cdot 2^{\Delta}
= \frac{\pi}{2^{\Delta-1}\,\Gamma(\Delta)} ,
\label{eq:legendre}
\eeq
reproduces the closed form distribution~\eqref{eq:wigner_bessel}.
\section{Direct evaluation of the normalized negativity}
\label{app:negativity}
This appendix records the evaluation of the seed-normalized negativity $\Neg(t)$ from the macroscopic Bessel envelope~\eqref{eq:wigner_bessel}, \emph{keeping the Bessel variable $Z=Q\tilde\theta$ throughout}. Two points are essential: the change of variables from the angular coordinate $\tilde\theta$ to $Z$ carries an explicit factor $1/Q$, and the discrete-to-continuum measure supplies three factors of two -- the mirror copy in the momentum sum~\eqref{eq:momentum_measure}, the even $\pm Z$ envelope in the angular integral~\eqref{eq:cDelta}, and the half-integer radial lattice~\eqref{eq:radial_measure}. Dropping any of them produces a spurious deficit of up to $8$ in the overall coefficient relative to the exact discrete evaluation.
\paragraph{Continuum limit.}
With $\tilde\theta=4\pi p/D+2\varphi_A$, the momentum sum becomes $\sum_p\to2\times\frac{D}{4\pi}\int_{\mathrm{one\ period}}\mathrm{d}\tilde\theta$ (mirror copy) and the radial sum $\sum_Q\to2\int_0^\infty\mathrm{d}Q$ (half-integer lattice), so
\beq
\Neg(t)\approx 4\cdot\frac{D}{4\pi}\int_0^\infty\!\mathrm{d}Q
\int_{\mathrm{one\ period}}\mathrm{d}\tilde\theta\,
\bigl|\Wnorm(Q,\tilde\theta,t)\bigr| ,
\qquad
\bigl|\Wnorm\bigr|=\frac{\pi\,Q^{2\Delta}}{D\,2^{\Delta-1}\Gamma(\Delta)}
e^{-2Q|\ln|A||}\frac{|J_\Delta(Z)|}{|Z|^\Delta} .
\eeq
Note the absence of any survival factor $\Psurv=|\psi_0|^2$: it was divided out at the level of the amplitudes~\eqref{eq:chi_def}.
\paragraph{Angular integral.}
Substituting $Z=Q\tilde\theta$, $\mathrm{d}\tilde\theta=\mathrm{d}Z/Q$, and using that the envelope is even in $Z$,
\beq
\int_{\mathrm{one\ period}}\mathrm{d}\tilde\theta\,\frac{|J_\Delta(Z)|}{|Z|^\Delta}
=\frac1Q\int_{-\infty}^{\infty}\mathrm{d}Z\,\frac{|J_\Delta(Z)|}{|Z|^\Delta}
=\frac{2c_\Delta}{Q} ,
\qquad c_\Delta=\int_0^\infty\frac{|J_\Delta(Z)|}{Z^\Delta}\mathrm{d}Z .
\label{eq:app_cDelta}
\eeq
The integral $c_\Delta$ is finite for $\Delta>1/2$: the integrand is regular at $Z=0$ ($\to2^{-\Delta}/\Gamma(\Delta+1)$) and decays as $Z^{-\Delta-1/2}$ at large $Z$. The naive large-$Z$ (Hankel) form alone would make the $\tilde\theta$-integral diverge at $\tilde\theta\to0$; that ``divergence'' is the region $Z\lesssim1$ where the Hankel form fails, and its proper treatment is precisely the finite $c_\Delta$ together with the factor $1/Q$.
\paragraph{Radial integral and assembly.}
The $1/Q$ reduces the radial power to $Q^{2\Delta-1}$:
\beq
\int_0^\infty\mathrm{d}Q\,Q^{2\Delta-1}e^{-2Q|\ln|A||}=\frac{\Gamma(2\Delta)}{(2|\ln|A||)^{2\Delta}} ,
\eeq
and collecting the coefficient,
$4\cdot\frac{D}{4\pi}\cdot\frac{\pi}{D\,2^{\Delta-1}\Gamma(\Delta)}\cdot 2c_\Delta
=\frac{2c_\Delta}{2^{\Delta-1}\Gamma(\Delta)}=\frac{c_\Delta}{2^{\Delta-2}\Gamma(\Delta)}$,
so that
\beq
\boxed{\;\Neg(t)\approx\frac{c_\Delta\,\Gamma(2\Delta)}{2^{\Delta-2}\Gamma(\Delta)}\,
\frac{1}{(2|\ln|A(t)||)^{2\Delta}}\;,}
\label{eq:app_neg_final}
\eeq
which is~\eqref{eq:neg_full}. Since $2|\ln|A||=-\ln(1-\xi)\simeq\xi\simeq\varepsilon_0/S^2$ at late times, $\Neg(t)\simeq\mathrm{const}\cdot\sinh^{4\Delta}(\pi t/\beta)/\sin^{4\Delta}(2\pi\varepsilon/\beta)$, the growth used in section~\ref{sec:growth}. The raw negativity is recovered as $\Neg_{\mathrm{raw}}=\Psurv\,\Neg=(c_\Delta\Gamma(2\Delta)/2^{\Delta-2}\Gamma(\Delta))\,[\xi/(-\ln(1-\xi))]^{2\Delta}$, which tends to the constant $\Neg_\infty=c_\Delta\Gamma(2\Delta)/2^{\Delta-2}\Gamma(\Delta)$ as $\xi\to0$: the normalized-state negativity saturates, and this coefficient -- not one computed on a single period, a single side of $Z$, or the unit radial lattice -- is the plateau value verified against the exact discrete Wigner function in figure~\ref{fig:raw_saturation}.
\newpage
\bibliographystyle{apsrev4-2}
\bibliography{references}
\end{document}